\documentclass[final, twocolumn]{IEEEtran}

\ifCLASSINFOpdf
\else
\fi

\usepackage[cmex10]{amsmath}

\usepackage{amssymb,amsthm}
\usepackage{color}
\allowdisplaybreaks

\makeatletter
\newtheorem*{rep@theorem}{\rep@title}
\newcommand{\newreptheorem}[2]{%
\newenvironment{rep#1}[1]{%
 \def\rep@title{#2 \ref{##1}}%
 \begin{rep@theorem}}%
 {\end{rep@theorem}}}
\makeatother

\theoremstyle{definition}
\newtheorem{definition}{Definition}
\newtheorem{construction}{Construction}
\newtheorem{example}[definition]{Example}

\theoremstyle{plain}
\newtheorem{theorem}{Theorem}
\newreptheorem{theorem}{Theorem}
\newtheorem{proposition}[definition]{Proposition}
\newtheorem{lemma}[definition]{Lemma}
\newtheorem{remark}[definition]{Remark}
\newtheorem{corollary}[definition]{Corollary}

\begin{document}
\title{Relative generalized matrix weights of matrix codes for universal security on wire-tap networks}
\author{Umberto~Mart{\'i}nez-Pe\~{n}as,~\IEEEmembership{Student Member,~IEEE,} %\IEEEmembership{Student Member,~IEEE,}
		and Ryutaroh~Matsumoto,~\IEEEmembership{Member,~IEEE,}%~\IEEEmembership{Member,~IEEE}% <-this % stops a space
\thanks{The work of U.~Mart{\'i}nez-Pe\~{n}as was supported by The Danish Council for Independent Research under Grant No. DFF-4002-00367 and Grant No. DFF-5137-00076B (``EliteForsk-Rejsestipendium''). The work of R.~Matsumoto was supported by the Japan Society for the Promotion of Science under Grant No. 26289116.}
\thanks{Parts of this paper have been presented at the 54th Annual Allerton Conference on Communication, Control, and Computing, Monticello, IL, USA, Oct 2016 \cite{allertonversion}.}
\thanks{U. Mart{\'i}nez-Pe\~{n}as is with the Department of Mathematical Sciences, Aalborg University, Aalborg 9220, Denmark (e-mail: umberto@math.aau.dk). At the time of submission, he was visiting The Edward S. Rogers Sr. Department of Electrical and Computer Engineering, University of Toronto, Toronto, ON M5S 3G4, Canada.}
\thanks{R. Matsumoto is with the Department of Information and Communication Engineering, Nagoya University, Nagoya, Japan.}}

%Copyright (c) 2014 IEEE. Personal use of this material is permitted. However, permission to use this material for any other purposes must be obtained from the IEEE by sending a request to pubs-permissions@ieee.org.}}% <-this % stops a space
% \thanks{Manuscript received April 19, 2005; revised September 17, 2014.}}


\markboth{ }%
{Shell \MakeLowercase{\textit{et al.}}: Bare Demo of IEEEtran.cls for IEEE Journals}



\maketitle

\begin{abstract}
Universal security over a network with linear network coding has been intensively studied. However, previous linear codes and code pairs used for this purpose were linear over a larger field than that used on the network{\color{black}, which restricts the possible packet lengths of optimal universal secure codes, does not allow to apply known list-decodable rank-metric codes and requires performing operations over a large field}. In this work, we introduce new parameters (relative generalized matrix weights and relative dimension/rank support profile) for code pairs that are linear over the field used in the network, {\color{black}and show that they measure} the universal security performance of these code pairs. For one code and non-square matrices, generalized matrix weights coincide with the existing Delsarte generalized weights, hence we prove the connection between these latter weights and secure network coding{\color{black}, which was left open}. {\color{black}As main applications,} the proposed new parameters enable us to{\color{black}: 1)} Obtain optimal universal secure linear codes on noiseless networks for all possible {\color{black}packet lengths, in particular for packet lengths not considered before, 2) Obtain the first universal secure list-decodable rank-metric code pairs with polynomial-sized lists, based on a recent construction} by Guruswami et al{\color{black}, and 3) Obtain new characterizations of security equivalences of linear codes}. Finally, we show that our parameters extend relative generalized Hamming weights and relative dimension/length profile, respectively, and relative generalized rank weights and relative dimension/intersection profile, respectively. \\
\end{abstract}


% \fi
\IEEEpeerreviewmaketitle

\section{Introduction}
\IEEEPARstart{L}{inear} network coding was first studied in \cite{ahlswede}, {\color{black}\cite{Koetter2003}} and \cite{linearnetwork}, and enables us to realize higher throughput than the conventional storing and forwarding. {\color{black}Error correction in this context was first studied in \cite{cai-yeung}, and security, }meaning information leakage to an adversary wire-tapping links in the network, was first considered in \cite{secure-network}. {\color{black}In that work, the authors give outer codes with optimal information rate for the given security performance, although {\color{black}using large fields on the network}. The field size was later reduced in \cite{feldman} by reducing the information rate.} In addition, the approach in \cite{wiretapnetworks} allows us to see {\color{black}secure network coding} as a generalization of secret sharing {\color{black}\cite{blakley-safe, shamir}}, which is a generalization of the wire-tap channel of type II \cite{ozarow}{\color{black}. }

However, {\color{black}these} approaches {\color{black}\cite{secure-network, wiretapnetworks, feldman}} require knowing and/or modifying the underlying linear network code, which does not allow us to perform, for instance, random linear network coding \cite{random}{\color{black}, which achieves capacity in a decentralized manner and is robust to network changes. }Later, the use of {\color{black}pairs of} linear (block) codes as outer codes was proposed in \cite{silva-universal} to protect messages from errors together with information leakage to a wire-tapping adversary {\color{black}(see Remark \ref{remark noiseless coset})}, depending only on the number of errors and wire-tapped links, respectively, and not depending on the underlying linear network code, which is referred to as \textit{universal security} in \cite{silva-universal}. 

{\color{black}In \cite{silva-universal}, the encoded message consists of $ n $ (number of outgoing links from the source) vectors in $ \mathbb{F}_{q^m} $ or $ \mathbb{F}_q^m $, called packets, where $ m $ is called the packet length and where $ \mathbb{F}_q $ is the field used for the underlying linear network code, as opposed to previous works \cite{secure-network, wiretapnetworks, feldman}, where $ m = 1 $. The universal performance of the proposed linear codes in \cite{silva-universal} is measured by the rank metric \cite{delsartebilinear}, and the authors in \cite{silva-universal} prove that linear codes in $ \mathbb{F}_{q^m}^n $ with optimal rank-metric parameters when $ n \leq m $ \cite{gabidulin, roth} are also optimal for universal security. This approach was already proposed in \cite{on-metrics, error-control} for error correction, again not depending on the underlying network code. Later the authors in \cite{list-decodable-rank-metric} obtained the first list-decodable rank-metric codes whose list sizes are polynomial in the code length and which are able to list-decode universally on linearly coded networks roughly twice as many errors as optimal rank-metric codes \cite{gabidulin, roth} can correct. The rank metric was then generalized} in \cite{rgrw} {\color{black}to} relative generalized rank weights (RGRWs) and relative dimension/intersection profiles (RDIPs){\color{black}, which were proven in \cite{rgrw}} to measure {\color{black}exactly and} simultaneously the universal security performance and error-correction capability of pairs of linear codes, {\color{black}in the same way as relative generalized Hamming weights (RGHWs) and relative dimension/length profiles (RDLPs) \cite{luo, wei} do on wire-tap channels of type II.}

Unfortunately, the codes studied and proposed in {\color{black}\cite{rgrw, on-metrics, silva-universal, error-control} for universal security} are linear over the extension field $ \mathbb{F}_{q^m} $. This restricts the {\color{black}possible packet lengths of optimal universal secure codes}, requires performing computations over the larger field $ \mathbb{F}_{q^m} $ and leaves out important codes, such as the {\color{black}list-decodable rank-metric codes in \cite{list-decodable-rank-metric}, which are only linear over $ \mathbb{F}_q $.} 

In this work, {\color{black}we introduce new parameters, called relative generalized matrix weights (RGMWs) and relative dimension/rank support profiles (RDRPs), for codes and code pairs that are linear over the smaller field $ \mathbb{F}_q $, and prove that they measure their universal security performance in terms of the worst-case information leakage. As main applications, we obtain the first optimal universal secure linear codes on noiseless networks for all possible packet lengths, we obtain the first universal secure list-decodable rank-metric code pairs with polynomial-sized lists, and obtain new characterizations of security equivalences of linear codes.}

{\color{black}
\subsection{Notation}

Let $ q $ be a prime power and $ m $ and $ n $, two positive integers. We denote by $ \mathbb{F}_q $ the finite field with $ q $ elements, which we will consider to be the field used for the underlying linear network code (see \cite[Definition 1]{Koetter2003}).

Most of our technical results hold for an arbitrary field, which we denote by $ \mathbb{F} $ and which mathematically plays the role of $ \mathbb{F}_q $. $ \mathbb{F}^n $ denotes the vector space of row vectors of length $ n $ with components in $ \mathbb{F} $, and $ \mathbb{F}^{m \times n} $ denotes the vector space of $ m \times n $ matrices with components in $ \mathbb{F} $. Throughout the paper, a (block) code in $ \mathbb{F}^{m \times n} $ (respectively, in $ \mathbb{F}^n $) is a subset of $ \mathbb{F}^{m \times n} $ (respectively, of $ \mathbb{F}^n $), and it is called linear if it is a vector space over $ \mathbb{F} $. In all cases, dimensions of vector spaces over $ \mathbb{F} $ will be denoted by $ \dim $.

Finally, we recall that we may identify $ \mathbb{F}_{q^m}^n $ and $ \mathbb{F}_q^{m \times n} $ as vector spaces over $ \mathbb{F}_q $. Fix a basis $ \alpha_1, \alpha_2, \ldots, \alpha_m $ of $ \mathbb{F}_{q^m} $ as a vector space over $ \mathbb{F}_q $. We define the \textit{matrix representation map} $ M_{\boldsymbol\alpha} : \mathbb{F}_{q^m}^n \longrightarrow \mathbb{F}_q^{m \times n} $ associated to the previous basis by
\begin{equation}
M_{\boldsymbol\alpha} (\mathbf{c}) = (c_{i,j})_{1 \leq i \leq m, 1 \leq j \leq n},
\label{equation matrix representation}
\end{equation}
where $ \mathbf{c}_i = (c_{i,1}, c_{i,2}, \ldots, c_{i,n}) \in \mathbb{F}_q^n $, for $ i = 1,2, \ldots, m $, are the unique vectors in $ \mathbb{F}_q^n $ such that $ \mathbf{c} = \sum_{i=1}^m \alpha_i \mathbf{c}_i $. The map $ M_{\boldsymbol\alpha} : \mathbb{F}_{q^m}^n \longrightarrow \mathbb{F}_q^{m \times n} $ is an $ \mathbb{F}_q $-linear vector space isomorphism. 

The works \cite{rgrw, on-metrics, silva-universal, error-control} consider $ \mathbb{F}_{q^m} $-linear codes in $ \mathbb{F}_{q^m}^n $, which are a subfamily of $ \mathbb{F}_q $-linear codes in $ \mathbb{F}_q^{m \times n} $ through the map given in (\ref{equation matrix representation}). In this paper, we will consider arbitrary linear (meaning $ \mathbb{F} $-linear) codes in $ \mathbb{F}^{m \times n} $.

\begin{table*}[!t]
\caption{New and existing notions of generalized weights}
\label{table notions gen weights}
\centering
\begin{tabular}{|c|c|c|c|}
\hline
Work & Paremeters & Codes they are used on & Measured security \\
\hline\hline
Definitions \ref{def RGMW} \& \ref{def RDRP}, \& Theorem \ref{maximum secrecy} & RGMW \& RDRP & $ \mathcal{C}_2 \varsubsetneqq \mathcal{C}_1 \subseteq \mathbb{F}_q^{m \times n} $, $ \mathbb{F}_q $-linear & Universal security on networks \\
\hline
\cite{rgrw, oggier} & RGRW \& RDIP & $ \mathcal{C}_2 \varsubsetneqq \mathcal{C}_1 \subseteq \mathbb{F}_{q^m}^n $, $ \mathbb{F}_{q^m} $-linear & Universal security on networks \\
\hline
\cite{ravagnaniweights} & DGW & $ \mathcal{C}_2 = \{ 0 \} \varsubsetneqq \mathcal{C}_1 \subseteq \mathbb{F}_q^{m \times n} $, $ \mathbb{F}_q $-linear & Universal security on networks for $ \mathbb{F}_{q^m} $-linear \\
\hline
\cite{luo, wei} & RGHW \& RDLP & $ \mathcal{C}_2 \varsubsetneqq \mathcal{C}_1 \subseteq \mathbb{F}_q^n $, $ \mathbb{F}_q $-linear & Security on wire-tap channels II \\
\hline
\cite{nghw, application-nghw} & RNGHW & $ \mathcal{C}_2 \varsubsetneqq \mathcal{C}_1 \subseteq \mathbb{F}_q^n $, $ \mathbb{F}_q $-linear & Non-universal security on networks \\
\hline
\end{tabular}
\end{table*}

\begin{table*}[!t]
\caption{New and existing optimal secure codes for noiseless networks ($ N = \# $ links, $ \mu = \# $ observations, $ t = \# $ destinations)}
\label{table optimal on noiseless}
\centering
\begin{tabular}{|c|c|c|c|}
\hline
Work & Universality & Field size ($ q $) used over the network & Packet length ($ m $) \\
\hline\hline
Theorem \ref{th optimal for noiseless networks} & Yes & Any & Any \\
\hline
\cite{silva-universal} & Yes & Any & $ m \geq n $ or $ n = lm $ \\
\hline
\cite{secure-network} & No & $ q > \binom{N}{\mu} $ & -- \\
\hline
\cite{feldman} & No & $ q = \Theta(N^{\mu / 2}) $ & -- \\
\hline
\cite{wiretapnetworks} & No & $ q > \binom{N-1}{\mu-1} + t $ or $ q > \binom{2n^3t^2 -1}{\mu-1} + t $ & -- \\
\hline
\end{tabular}
\end{table*}

\begin{table*}[!t]
\caption{New and existing characterizations of linear isomorphisms between vector spaces of matrices preserving certain properties}
\label{table isomorphism charact}
\centering
\begin{tabular}{|c|c|c|c|}
\hline
Work & Domain \& codomain & Linearity & Properties preserved \\
\hline\hline
Theorem \ref{theorem characterization} & $ \phi : \mathcal{V} \longrightarrow \mathcal{W} $, $ \mathcal{V} \in {\rm RS}(\mathbb{F}^{m \times n}) $, $ \mathcal{W} \in {\rm RS}(\mathbb{F}^{m \times n^\prime}) $ & $ \mathbb{F} $-linear & Universal security on networks \\
\hline
\cite{similarities} & $ \phi : \mathcal{V} \longrightarrow \mathcal{W} $, $ \mathcal{V} \in {\rm RS}(\mathbb{F}_q^{m \times n}) $, $ \mathcal{W} \in {\rm RS}(\mathbb{F}_q^{m \times n^\prime}) $ & $ \mathbb{F}_{q^m} $-linear & Ranks \& universal security on networks \\
\hline
\cite{berger} & $ \phi : \mathbb{F}_{q^m}^n \longrightarrow \mathbb{F}_{q^m}^n $ & $ \mathbb{F}_{q^m} $-linear & Ranks \\
\hline
\cite{marcus, morrison} & $ \phi : \mathbb{F}^{m \times n} \longrightarrow \mathbb{F}^{m \times n} $ & $ \mathbb{F} $-linear & Ranks, determinants \& eigenvalues \\
\hline
\cite{dieudonne} & $ \phi : \mathbb{F}^{n \times n} \longrightarrow \mathbb{F}^{n \times n} $ & $ \mathbb{F} $-linear & Invertible matrices \\
\hline
\end{tabular}
\end{table*}

\subsection{Our motivations}

Our main motivation to study universal secure network coding is to avoid knowing and/or modifying the underlying linear network code, and in particular be able to apply our theory on random linearly coded networks \cite{random}, which achieve capacity in a decentralized manner and are robust to network changes. 

Our main motivation to study pairs of linear codes is to be able to protect messages simultaneously from errors, erasures and information leakage to a wire-tapper. See also Section \ref{sec preliminaries} and more concretely, Remark \ref{remark noiseless coset}.

Our main motivations to study codes which are linear over the base field $ \mathbb{F}_q $ instead of the extension field $ \mathbb{F}_{q^m} $ are the following:

1) $ \mathbb{F}_q $-linear codes with optimal rank-metric parameters \cite{delsartebilinear}, and thus with optimal universal security and error-correction capability, cannot be $ \mathbb{F}_{q^m} $-linear for most packet lengths $ m $ when $ m < n $. In many applications, packet lengths satisfying $ m < n $ are required (see the discussion in \cite[Subsection I-A]{rgrw}, for instance).

2) The only known list-decodable rank-metric codes \cite{list-decodable-rank-metric} with polynomial-sized lists are linear over $ \mathbb{F}_q $, but not over $ \mathbb{F}_{q^m} $. Hence the previous studies on universal security cannot be applied on these codes. In particular, no construction of universal secure list-decodable rank-metric coding schemes with polynomial-sized lists are known.

3) In previous works \cite{on-metrics, silva-universal, error-control}, the proposed codes are $ \mathbb{F}_{q^m} $-linear and $ m \geq n $. In many cases, this requires performing operations over a very large field, instead of the much smaller field $ \mathbb{F}_q $.

\subsection{Related works and considered open problems}

We consider the following four open problems in the literature, which correspond to the main four contributions listed in the following subsection:

1) Several parameters have been introduced to measure the security performance of linear codes and code pairs on different channels, in terms of the worst-case information leakage. The original RGHWs and RDLPs \cite{luo, wei} measure security performance over wire-tap channels of type II, and relative network generalized Hamming weigths (RNGHWs) \cite{nghw, application-nghw} measure security performance over networks depending on the underlying linear network code (non-universal security). Later, RGRWs and RDIPs were introduced in \cite{rgrw, oggier} to measure universal security performance of $ \mathbb{F}_{q^m} $-linear code pairs $ \mathcal{C}_2 \varsubsetneqq \mathcal{C}_1 \subseteq \mathbb{F}_{q^m}^n $. A notion of generalized weight for one $ \mathbb{F}_q $-linear code (that is, for an arbitrary $ \mathcal{C}_1 $ and for $ \mathcal{C}_2 = \{ 0 \} $) in $ \mathbb{F}_q^{m \times n} $, called Delsarte generalized weights (DGWs), was introduced in \cite{ravagnaniweights}, but its connection with universal security was only given for $ \mathbb{F}_{q^m} $-linear codes. Thus, no measure of universal security performance for all $ \mathbb{F}_q $-linear codes or code pairs is known. See also Table \ref{table notions gen weights}.

2) The first optimal universal secure linear codes for noiseless networks were obtained in \cite[Section V]{silva-universal}, whose information rate attain the information-theoretical limit given in \cite{secure-network}. However, these codes only exist when $ m \geq n $. The cartesian products in \cite[Subsection VII-C]{silva-universal} are also optimal among $ \mathbb{F}_q $-linear codes (see Remark \ref{remark other optimal on noiseless}), but only exist when $ m $ divides $ n $. No optimal universal secure $ \mathbb{F}_q $-linear codes for noiseless networks have been obtained for the rest of values of $ m $. See also Table \ref{table optimal on noiseless} for an overview of existing optimal constructions, including non-universal codes \cite{secure-network, wiretapnetworks, feldman}.

3) In \cite{list-decodable-rank-metric}, the authors introduce the first list-decodable rank-metric codes in $ \mathbb{F}_{q^m}^n $ able to list-decode close to the information-theoretical limit and roughly twice as many errors as optimal rank-metric codes \cite{gabidulin, roth} are able to correct, in polynomial time and with polynomial-sized lists (on the length $ n $). However, no universal secure coding schemes with such list-decoding capabilities are known. Observe that list-decoding rank errors implies list-decoding errors in linear network coding in a universal manner \cite{on-metrics}.

4) Several characterizations of maps between vector spaces of matrices preserving certain properties have been given in the literature \cite{berger, dieudonne, marcus, similarities, morrison}. The maps considered in \cite{berger} are linear over the extension field $ \mathbb{F}_{q^m} $ and preserve ranks, and the maps considered in \cite{dieudonne, marcus, morrison} are linear over the base field ($ \mathbb{F}_q $ or an arbitrary field) and preserve fundamental properties of matrices, such as ranks, determinants, eigenvalues or invertible matrices. Characterizations of maps preserving universal security performance were first given in \cite{similarities}, although the considered maps were only linear over $ \mathbb{F}_{q^m} $. No characterizations of general $ \mathbb{F}_q $-linear maps preserving universal security are known. See also Table \ref{table isomorphism charact}.

\subsection{Our contributions and main results}

In the following, we list our four main contributions together with our main result summarizing each of them. Each contribution tackles each open problem listed in the previous subsection, respectively.

1) We introduce new parameters, RGMWs and RDRPs, in Definitions \ref{def RGMW} and \ref{def RDRP}, respectively, which measure the universal security performance of $ \mathbb{F}_q $-linear code pairs $ \mathcal{C}_2 \varsubsetneqq \mathcal{C}_1 \subseteq \mathbb{F}_q^{m \times n} $, in terms of the worst-case information leakage. The main result is Theorem \ref{maximum secrecy} and states the following: The $ r $-th RGMW of the code pair is the minimum number of links that an adversary needs to wire-tap in order to obtain at least $ r $ bits of information (multiplied by $ \log_2(q) $) about the sent message. The $ \mu $-th RDRP of the code pair is the maximum number of bits of information (multiplied by $ \log_2(q) $) about the sent message that can be obtained by wire-tapping $ \mu $ links of the network. 

Since $ \mathbb{F}_{q^m} $-linear codes in $ \mathbb{F}_{q^m}^n $ are also $ \mathbb{F}_q $-linear codes in $ \mathbb{F}_q^{m \times n} $, RGMWs and RDRPs must coincide with RGRWs and RDIPs \cite{rgrw}, respectively, for $ \mathbb{F}_{q^m} $-linear codes in $ \mathbb{F}_{q^m}^n $, which we prove in Theorem \ref{th RGMW extend RGRW}.

When $ \mathcal{C}_2 = \{ 0 \} $ and $ m \neq n $, we will also show in Theorem \ref{th GMW extend DGW} that the RGMWs of the pair coincide with their DGWs, given in \cite{ravagnaniweights}, hence proving the connection between DGWs and universal security for general $ \mathbb{F}_q $-linear codes, which was left open.

2) We obtain optimal universal secure $ \mathbb{F}_q $-linear codes for noiseless networks for any value of $ m $ and $ n $, not only when $ m \geq n $ or $ m $ divides $ n $, as in previous works \cite{silva-universal}. The main result is Theorem \ref{th optimal for noiseless networks}, which states the following: Denote by $ \ell $ the number of packets in $ \mathbb{F}_q^m $ that the source can transmit and by $ t $ the number of links the adversary may wire-tap without obtaining any information about the sent packets. For any $ m $ and $ n $, and a fixed value of $ \ell $ (respectively $ t $), we obtain a coding scheme with optimal value of $ t $ (respectively $ \ell $).

3) We obtain the first universal secure list-decodable rank-metric code pairs with polynomial-sized lists. The main result is Theorem \ref{th secure list-decodable}, and states the following: Defining $ \ell $ and $ t $ as in the previous item, assuming that $ n $ divides $ m $, and fixing $ 1 \leq k_2 < k_2 \leq n $, $ \varepsilon > 0 $ and a positive integer $ s $ such that $ 4sn \leq \varepsilon m $ and $ m/n = \mathcal{O}(s/ \varepsilon) $, we obtain an $ \mathbb{F}_q $-linear code pair such that $ \ell \geq m(k_1 - k_2)(1 - 2\varepsilon) $, $ t \geq k_2 $ and which can list-decode $ \frac{s}{s+1}(n-k_1) $ rank errors in polynomial time, where the list size is $ q^{\mathcal{O}(s^2 / \varepsilon^2)} $.

4) We obtain characterizations of vector space isomorphisms between certain spaces of matrices over $ \mathbb{F}_q $ that preserve universal security performance over networks. The main result is Theorem \ref{theorem characterization}, which gives several characterizations of $ \mathbb{F}_q $-linear vector space isomorphisms $ \phi : \mathcal{V} \longrightarrow \mathcal{W} $, where $ \mathcal{V} $ and $ \mathcal{W} $ are rank support spaces in $ \mathbb{F}_q^{m \times n} $ and $ \mathbb{F}_q^{m \times n^\prime} $ (see Definition \ref{def rank support spaces}), respectively. 

As application, we obtain in Subsection \ref{subsec minimum parameters} ranges of possible parameters $ m $ and $ n $ that given linear codes and code pairs can be applied to without changing their universal security performance.

}

\subsection{Organization of the paper}

{\color{black}First, all of our main results are stated as \textit{Theorems}.} After some preliminaries in Section II, we introduce in Section III the new parameters of linear code pairs (RGMWs and RDRPs), give their connection {\color{black}with the rank metric}, and prove that they exactly measure the worst-case information leakage universally on networks {\color{black}(Theorem \ref{maximum secrecy})}. In Section IV, we give optimal universal secure linear codes for noiseless networks for all possible parameters {\color{black}(Theorem \ref{th optimal for noiseless networks})}. In Section V, we show how to add universal security to the list-decodable rank-metric codes in \cite{list-decodable-rank-metric} {\color{black}(Theorem \ref{th secure list-decodable})}. In Section VI, we define and give characterizations of security equivalences of linear codes {\color{black}(Theorem \ref{theorem characterization})}, and then obtain ranges of possible parameters of linear codes up to these equivalences. In Section VII, we give upper and lower Singleton-type bounds {\color{black}(Theorems \ref{th upper singleton} and \ref{lower singleton bound})} and study when they can be attained, when the dimensions are divisible by $ m $. Finally, in Section VIII, we prove that RGMWs extend RGRWs \cite{rgrw} and RGHWs \cite{luo, wei}, and we prove that RDRPs extend RDIPs \cite{rgrw} and RDLPs \cite{forney, luo} {\color{black}(Theorems \ref{th RGMW extend RGRW} and \ref{th RGMW extend RGHW}, respectively)}. We conclude the section by showing that GMWs coincide with DGWs \cite{ravagnaniweights} for non-square matrices, and are strictly larger otherwise {\color{black}(Theorem \ref{th GMW extend DGW})}. 

\section{Coset coding schemes for universal security in linear network coding} \label{sec preliminaries}

{\color{black}This section serves as a brief summary of the model of linear network coding that we consider (Subsection \ref{subsec linear network model}), the concept of universal security under this model (Subsection \ref{subsection secure communication}) and the main definitions concerning coset coding schemes used for this purpose (Subsection \ref{subsection coding schemes}). The section only contains definitions and facts known in the literature, which will be used throughout the paper.
}

\subsection{Linear network coding model} \label{subsec linear network model}

Consider a network with several sources and several sinks. A given source transmits a message $ \mathbf{x} \in \mathbb{F}_q^\ell $ through the network to multiple sinks. To that end, that source encodes the message as a collection of $ n $ packets of length $ m $, seen as a matrix $ C \in \mathbb{F}_q^{m \times n} $, where $ n $ is the number of outgoing links from this source. We consider linear network coding on the network, first considered in \cite{ahlswede, linearnetwork} and formally defined in \cite[Definition 1]{Koetter2003}, which allows us to reach higher throughput than just storing and forwarding on the network. This means that a given sink receives a matrix of the form
\begin{equation*}
Y = CA^T \in \mathbb{F}_q^{m \times N},
%\label{linear network equation}
\end{equation*}
where $ A \in \mathbb{F}_q^{N \times n} $ is called the transfer matrix corresponding to the considered source and sink{\color{black}, and $ A^T $ denotes its transpose}. This matrix may be randomly chosen if random linear network coding is applied \cite{random}.

\subsection{Universal secure communication over networks} \label{subsection secure communication}

In secure {\color{black}and} reliable network coding, two of the main problems addressed in the literature are the following:
\begin{enumerate}
\item
Error and erasure correction \cite{cai-yeung, rgrw, on-metrics, silva-universal, error-control}: An adversary and/or a noisy channel may introduce errors on some links of the network and/or modify the transfer matrix. In this case, the sink receives the matrix 
$$ Y = CA^{\prime T} + E \in \mathbb{F}_q^{m \times N}, $$
where $ A^\prime \in \mathbb{F}_q^{N \times n} $ is the modified transfer matrix, and $ E \in \mathbb{F}_q^{m \times N} $ is the final error matrix. In this case, we say that $ t = {\rm Rk}(E) $ errors and $ \rho = n - {\rm Rk}(A^\prime) $ erasures occurred{\color{black}, where {\rm Rk} denotes the rank of a matrix}.
\item
Information leakage \cite{secure-network, wiretapnetworks, feldman, rgrw, silva-universal}: A wire-tapping adversary listens to $ \mu > 0 $ links of the network, obtaining a matrix of the form $ CB^T \in \mathbb{F}_q^{m \times \mu} $, for some matrix $ B \in \mathbb{F}_q^{\mu \times n} $. %When considering information leakage, we assume that the messages sent by different sources have no correlations, since then they do not give information about each other on the network by \cite[Proposition 5]{rgrw}.
\end{enumerate}

Outer coding in the source node is usually applied to tackle the previous problems, and it is called \textit{universal secure} \cite{silva-universal} if it provides reliability and security as in the previous items for fixed numbers of wire-tapped links $ \mu $, errors $ t $ and erasures $ \rho $, independently of the transfer matrix $ A $ used. This implies that no previous knowledge or modification of the transfer matrix is required and random linear network coding \cite{random} may be applied.

\subsection{Coset coding schemes for outer codes} \label{subsection coding schemes}

Coding techniques for protecting messages simultaneously from errors and information leakage to a wire-tapping adversary were first studied by Wyner in \cite{wyner}. In \cite[p. 1374]{wyner}, the general concept of coset coding scheme, as we will next define, was first introduced for this purpose. We use the formal definition in \cite[Definition 7]{rgrw}:

\begin{definition}[\textbf{Coset coding schemes \cite{rgrw, wyner}}]
A coset coding scheme over the field $ \mathbb{F} $ with message set $ \mathcal{S} $ is a family of disjoint nonempty subsets of $ \mathbb{F}^{m \times n} $, $ \mathcal{P}_\mathcal{S} = \{ \mathcal{C}_\mathbf{x} \}_{\mathbf{x} \in \mathcal{S}} $.

{\color{black}If $ \mathbb{F} = \mathbb{F}_q $,} each $ \mathbf{x} \in \mathcal{S} $ is encoded by the source by choosing uniformly at random an element $ C \in \mathcal{C}_\mathbf{x} $. 
\end{definition}

\begin{definition}[\textbf{Linear coset coding schemes \cite[Definition 2]{similarities}}]
A coset coding scheme as in the previous definition is said to be linear if $ \mathcal{S} = \mathbb{F}^\ell $, for some $ 0 < \ell \leq mn $, and
$$ a \mathcal{C}_\mathbf{x} + b \mathcal{C}_\mathbf{y} \subseteq \mathcal{C}_{a \mathbf{x} + b \mathbf{y}}, $$
for all $ a, b \in \mathbb{F} $ and all $ \mathbf{x}, \mathbf{y} \in \mathbb{F}^\ell $. 
\end{definition}

With these definitions, the concept of coset coding scheme generalizes the concept of (block) code, since a code is a coset coding scheme where $ | \mathcal{C}_\mathbf{x} | = 1 $, for each $ \mathbf{x} \in \mathcal{S} $. In the same way, linear coset coding schemes generalize linear (block) codes.

An equivalent way to describe linear coset coding schemes is by nested linear code pairs, introduced in \cite[Section III.A]{zamir}. We use the description in \cite[Subsection 4.2]{secure-computation}.

\begin{definition}[\textbf{Nested linear code pairs \cite{secure-computation, zamir}}] \label{definition NLCP}
A nested linear code pair is a pair of linear codes $ \mathcal{C}_2 \subsetneqq \mathcal{C}_1 \subseteq \mathbb{F}^{m \times n} $. Choose a vector space $ \mathcal{W} $ such that $ \mathcal{C}_1 = \mathcal{C}_2 \oplus \mathcal{W} ${\color{black}, where $ \oplus $ denotes the direct sum of vector spaces,} and a vector space isomorphism $ \psi : \mathbb{F}^\ell \longrightarrow \mathcal{W} $, where $ \ell = \dim(\mathcal{C}_1/\mathcal{C}_2) $. Then we define $ \mathcal{C}_\mathbf{x} = \psi(\mathbf{x}) + \mathcal{C}_2 $, for $ \mathbf{x} \in \mathbb{F}^\ell $. They form a linear coset coding scheme called nested coset coding scheme \cite{rgrw}.
\end{definition}
 
\begin{remark} \label{remark noiseless coset}
As observed in \cite{ozarow} for the wire-tap channel of type II, linear code pairs {\color{black}where $ \mathcal{C}_1 = \mathbb{F}^{m \times n} $} are suitable for protecting information from leakage on noiseless channels. {\color{black}Analogously, linear code pairs where $ \mathcal{C}_2 = \{ 0 \} $ are suitable for error correction without the presence of eavesdroppers. Observe that these two types of linear code pairs are \textit{dual} to each other (see Definition \ref{def trace inner product} and Appendix \ref{app duality theory}): If $ \mathcal{C}_1^\prime = \mathcal{C}_2^\perp $ and $ \mathcal{C}_2^\prime = \mathcal{C}_1^\perp $, then $ \mathcal{C}_1 = \mathbb{F}^{m \times n} $ if, and only if, $ \mathcal{C}_2^\prime = \{ 0 \} $. To treat both error correction and information leakage, we need general linear coset coding schemes.}
\end{remark}

We recall here that the concept of linear coset coding schemes and nested coset coding schemes are exactly the same. An object in the first family uniquely defines an object in the second family and vice-versa. This is formally proven in \cite[Proposition 1]{similarities}.

Finally, we recall that the exact universal error and erasure correction capability of a nested coset coding scheme was found, in terms of the rank metric, first in \cite[Section IV.C]{on-metrics} for the case of one code {\color{black}($ \mathcal{C}_2 = \{ 0 \} $)} that is maximum rank distance, then in \cite[Theorem 2]{silva-universal} for the general case of one linear code {\color{black}(again $ \mathcal{C}_2 = \{ 0 \} $)}, then in \cite[Theorem 4]{rgrw} for the case where both codes are linear over an extension field $ \mathbb{F}_{q^m} $, and finally in \cite[Theorem 9]{similarities} for arbitrary coset coding schemes (linear over $ \mathbb{F}_q $ and non-linear).

\section{New parameters of linear coset coding schemes for universal security on networks} \label{sec introduction of RGMWs}

{\color{black}This is the main section of the paper, which serves as a basis for the rest of sections. The next sections can be read independently of each other, but all of them build on the results in this section. Here we introduce rank support spaces (Subsection \ref{subsec rank supports}), which are the main technical building blocks of our theory, then we define of our main parameters and connect them with the rank metric (Subsection \ref{subsec def of RGMWs}), and we conclude by showing (Theorem \ref{maximum secrecy}) that these parameters measure the worst-case information leakage universally on linearly coded networks (Subsection \ref{subsec measuring info leakage}). }

\subsection{Rank supports and rank support spaces} \label{subsec rank supports}

{\color{black}In this subsection, we introduce rank support spaces, which are the mathematical building blocks of our theory. The idea is to attach to each linear code its rank support, given in \cite[Definition 1]{slides}, and based on this rank support, define a vector space of matrices containing the original code that can be seen as its ambient space with respect to the rank metric. 

We remark here that the family of rank support spaces can be seen as the family of vector spaces in \cite[Notation 25]{ravagnani} after transposition of matrices, or the family of vector spaces in \cite[Definition 6]{slides} taking $ \mathcal{C} = \mathbb{F}_q^{m \times n} $. We start with the definitions: }
 
\begin{definition} [\textbf{Row space and rank}]
For a matrix $ C \in \mathbb{F}^{m \times n} $, we define its row space $ {\rm Row}(C) $ as the vector space in $ \mathbb{F}^n $ generated by its rows. As usual, we define its rank as $ {\rm Rk}(C) = \dim({\rm Row}(C)) $.
\end{definition}

\begin{definition} [\textbf{Rank support and rank weight \cite[Definition 1]{slides}}] 
Given a vector space $ \mathcal{C} \subseteq \mathbb{F}^{m \times n} $, we define its rank support as
\begin{equation*}
{\rm RSupp}(\mathcal{C}) = \sum_{C \in \mathcal{C}} {\rm Row}(C) \subseteq \mathbb{F}^n.
%\label{rank support}
\end{equation*}
We also define the rank weight of the space $ \mathcal{C} $ as 
\begin{equation*}
{\rm wt_R}(\mathcal{C}) = \dim({\rm RSupp}(\mathcal{C})).
%\label{rank weigth}
\end{equation*}
\end{definition}

{\color{black}Observe that} $ {\rm RSupp}(\langle \{ C \} \rangle) = {\rm Row}(C) $ and $ {\rm wt_R}(\langle \{ C \} \rangle) $ $ = {\rm Rk}(C) $, for every matrix $ C \in \mathbb{F}^{m \times n} ${\color{black}, where $ \langle \mathcal{A} \rangle $ denotes the vector space generated by a set $ \mathcal{A} $ over $ \mathbb{F} $.} 

\begin{definition} [\textbf{Rank support spaces}] \label{def rank support spaces}
Given a vector space $ \mathcal{L} \subseteq \mathbb{F}^n $, we define its rank support space $ \mathcal{V}_\mathcal{L} \subseteq \mathbb{F}^{m \times n} $ as
\begin{equation*}
\mathcal{V}_\mathcal{L} = \{ V \in \mathbb{F}^{m \times n} \mid {\rm Row}(V) \subseteq \mathcal{L} \}.
%\label{rank support space}
\end{equation*}
We denote by $ RS(\mathbb{F}^{m \times n}) $ the family of rank support spaces in $ \mathbb{F}^{m \times n} $.
\end{definition}

{\color{black}The following lemma shows that rank support spaces behave as a sort of ambient spaces for linear codes and can be attached bijectively to vector spaces in $ \mathbb{F}^n $, which correspond to the rank supports of the original linear codes. }

\begin{lemma} \label{basic lemma matrix modules}
Let $ \mathcal{L} \subseteq \mathbb{F}^n $ be a vector space. The following hold:
\begin{enumerate}
\item
$ \mathcal{V}_\mathcal{L} $ is a vector space and the correspondence $ \mathcal{L} \mapsto \mathcal{V}_\mathcal{L} $ between subspaces of $ \mathbb{F}^n $ and rank support spaces is a bijection with inverse $ \mathcal{V}_\mathcal{L} \mapsto {\rm RSupp}(\mathcal{V}_\mathcal{L}) = \mathcal{L} $.
\item
If $ \mathcal{C} \subseteq \mathbb{F}^{m \times n} $ is a vector space and $ \mathcal{L} = {\rm RSupp}(\mathcal{C}) $, then $ \mathcal{V}_\mathcal{L} $ is the smallest rank support space containing $ \mathcal{C} $.
\end{enumerate}
\end{lemma}

{\color{black}We conclude the subsection with the following characterizations of rank support spaces, which we will use throughout the paper. In particular, item 2 will be useful to prove Theorem \ref{theorem characterization}, and item 3 will be useful to prove Theorem \ref{maximum secrecy}. }

\begin{proposition} \label{prop characterization}
Fix a set $ \mathcal{V} \subseteq \mathbb{F}^{m \times n} $. The following are equivalent:
\begin{enumerate}
\item
$ \mathcal{V} $ is a rank support space. That is, there exists a subspace $ \mathcal{L} \subseteq \mathbb{F}^n $ such that $ \mathcal{V} = \mathcal{V}_\mathcal{L} $.
\item
$ \mathcal{V} $ is linear and has a basis of the form $ B_{i,j} $, for $ i = 1,2, \ldots, m $ and $ j = 1,2, \ldots, k $, where there are vectors $ \mathbf{b}_1, \mathbf{b}_2, \ldots, \mathbf{b}_k \in \mathbb{F}^n $ such that $ B_{i,j} $ has the vector $ \mathbf{b}_j $ in the $ i $-th row and the rest of its rows are zero vectors.
\item
There exists a matrix $ B \in \mathbb{F}^{\mu \times n} $, for some positive integer $ \mu $, such that
$$ \mathcal{V} = \{ V \in \mathbb{F}^{m \times n} \mid VB^T = 0 \}. $$
\end{enumerate}
In addition, the relation between items 1, 2 and 3 is that $ \mathbf{b}_1, \mathbf{b}_2, \ldots, \mathbf{b}_k $ are a basis of $ \mathcal{L} $, $ B $ is a (possibly not full-rank) parity check matrix of $ \mathcal{L} $ and $ \dim(\mathcal{L}) = n - {\rm Rk}(B) $. 

In particular, it holds that
\begin{equation}
\dim(\mathcal{V}_\mathcal{L}) = m \dim(\mathcal{L}).
\label{dimension matrix modules}
\end{equation}
\end{proposition}
\begin{proof}
We prove the following implications:
\begin{itemize}
\item
$ 1 \Longleftrightarrow 2 $: Assume item 1, let $ \mathbf{b}_1, \mathbf{b}_2, \ldots, \mathbf{b}_k $ be a basis of $ \mathcal{L} $, and let $ B_{i,j} $ be as in item 2. {\color{black}Then we} see that $ \mathcal{V} = \langle \{ B_{i,j} \mid 1 \leq i \leq m, 1 \leq j \leq k \} \rangle $. The reversed implication follows in the same way by defining $ \mathcal{L} = \langle \mathbf{b}_1, \mathbf{b}_2, \ldots, \mathbf{b}_k \rangle \subseteq \mathbb{F}^n $.
\item
$ 1 \Longleftrightarrow 3 $: Assume item 1 and let $ B \in \mathbb{F}^{\mu \times n} $ be a parity check matrix of $ \mathcal{L} $. That is, a generator matrix of the dual $ \mathcal{L}^\perp \subseteq \mathbb{F}^n $. Then it holds by definition that $ V \in \mathbb{F}^{m \times n} $ has all its rows in $ \mathcal{L} $ if, and only if, $ VB^T = 0 $. Conversely, assuming item 3 and defining $ \mathcal{L} $ as the code with parity check matrix $ B $, we see that $ \mathcal{V} = \mathcal{V}_\mathcal{L} $ by the same argument. Hence the result follows.
\end{itemize}
\end{proof}

{\color{black}
\subsection{Definition and basic properties of the new parameters} \label{subsec def of RGMWs}
}

\begin{definition} [\textbf{Relative Generalized Matrix Weight}] \label{def RGMW}
Given nested linear codes $ \mathcal{C}_2 \subsetneqq \mathcal{C}_1 \subseteq \mathbb{F}^{m \times n} $, and $ 1 \leq r \leq \ell = \dim(\mathcal{C}_1 / \mathcal{C}_2) $, we define their $ r $-th relative generalized matrix weight (RGMW) as
\begin{equation*}
\begin{split}
d_{M,r}(\mathcal{C}_1, \mathcal{C}_2) = \min \{ & \dim(\mathcal{L}) \mid \mathcal{L} \subseteq \mathbb{F}^n, \\
 & \dim(\mathcal{C}_1 \cap \mathcal{V}_\mathcal{L}) - \dim(\mathcal{C}_2 \cap \mathcal{V}_\mathcal{L}) \geq r \}.
\end{split}
%\label{RGMW}
\end{equation*}
For a linear code $ \mathcal{C} \subseteq \mathbb{F}^{m \times n} $, and $ 1 \leq r \leq \dim(\mathcal{C}) $, we define its $ r $-th generalized matrix weight (GMW) as
\begin{equation}
d_{M,r}(\mathcal{C}) = d_{M,r}(\mathcal{C},\{ 0 \}).
\label{GMW}
\end{equation}
\end{definition}

Observe that it holds that
\begin{equation}
d_{M,r}(\mathcal{C}_1, \mathcal{C}_2) \geq d_{M,r}(\mathcal{C}_1),
\label{RGMW higher than GMW}
\end{equation}
for all nested linear codes $ \mathcal{C}_2 \subsetneqq \mathcal{C}_1 \subseteq \mathbb{F}^{m \times n} $, and all $ 1 \leq r \leq \ell = \dim(\mathcal{C}_1 / \mathcal{C}_2) $.

\begin{definition} [\textbf{Relative Dimension/Rank support Profile}] \label{def RDRP}
Given nested linear codes $ \mathcal{C}_2 \subsetneqq \mathcal{C}_1 \subseteq \mathbb{F}^{m \times n} $, and $ 0 \leq \mu \leq n $, we define their $ \mu $-th relative dimension/rank support profile (RDRP) as
\begin{equation*}
\begin{split}
K_{M,\mu}(\mathcal{C}_1, \mathcal{C}_2) = \max \{ & \dim(\mathcal{C}_1 \cap \mathcal{V}_\mathcal{L}) - \dim(\mathcal{C}_2 \cap \mathcal{V}_\mathcal{L}) \mid \\
 & \mathcal{L} \subseteq \mathbb{F}^n, \dim(\mathcal{L}) \leq \mu \}.
\end{split}
%\label{RDIP}
\end{equation*}
\end{definition}

{\color{black}Now, if $ \mathcal{U} \subseteq \mathcal{V} \subseteq \mathbb{F}^{m \times n} $ are vector spaces, the natural linear map $ \mathcal{C}_1 \cap \mathcal{U} / \mathcal{C}_2 \cap \mathcal{U} \longrightarrow \mathcal{C}_1 \cap \mathcal{V} / \mathcal{C}_2 \cap \mathcal{V} $ is one to one. Therefore,} since we are taking maximums, it holds that
\begin{equation*}
\begin{split}
K_{M,\mu}(\mathcal{C}_1, \mathcal{C}_2) = \max \{ & \dim(\mathcal{C}_1 \cap \mathcal{V}_\mathcal{L}) - \dim(\mathcal{C}_2 \cap \mathcal{V}_\mathcal{L}) \mid \\
 & \mathcal{L} \subseteq \mathbb{F}^n, \dim(\mathcal{L}) = \mu \}.
\end{split}
\end{equation*}

We remark here that some existing notions of relative generalized weights from the literature are particular cases of RGMWs. The corresponding connections are given in Section \ref{sec relation with others}. In particular, GMWs of one linear code coincide with DGWs (introduced in \cite{ravagnaniweights}) for non-square matrices.

We next obtain the following characterization of RGMWs that gives an analogous description to the original definition of GHWs by Wei \cite{wei}:

\begin{proposition} \label{proposition as minimum rank weights}
Given nested linear codes $ \mathcal{C}_2 \subsetneqq \mathcal{C}_1 \subseteq \mathbb{F}^{m \times n} $, and an integer $ 1 \leq r \leq \dim(\mathcal{C}_1 / \mathcal{C}_2) $, it holds that
\begin{equation*}
\begin{split}
d_{M,r}(\mathcal{C}_1, \mathcal{C}_2) = \min \{ & {\rm wt_R}(\mathcal{D}) \mid \mathcal{D} \subseteq \mathcal{C}_1, \mathcal{D} \cap \mathcal{C}_2 = \{ 0 \}, \\
 & \dim(\mathcal{D}) = r \}.
\end{split}
%\label{RDGW alternative definition}
\end{equation*}
\end{proposition}
\begin{proof}
Denote by $ d_r $ the number on the left-hand side and by $ d^\prime_r $ the number on the right-hand side. We prove both inequalities:

$ d_r \leq d^\prime_r $: Take a vector space $ \mathcal{D} \subseteq \mathcal{C}_1 $ such that $ \mathcal{D} \cap \mathcal{C}_2 = \{ 0 \} $, $ \dim(\mathcal{D}) = r $ and $ {\rm wt_R}(\mathcal{D}) = d^\prime_r $. Define $ \mathcal{L} = {\rm RSupp}(\mathcal{D}) $.

Since $ \mathcal{D} \subseteq \mathcal{V}_\mathcal{L} $, we have that $ \dim((\mathcal{C}_1 \cap \mathcal{V}_\mathcal{L})/(\mathcal{C}_2 \cap \mathcal{V}_\mathcal{L})) \geq \dim((\mathcal{C}_1 \cap \mathcal{D})/(\mathcal{C}_2 \cap \mathcal{D})) = \dim(\mathcal{D}) = r $. Hence
$$ d_r \leq \dim(\mathcal{L}) = {\rm wt_R}(\mathcal{D}) = d^\prime_r. $$

$ d_r \geq d^\prime_r $: Take a vector space $ \mathcal{L} \subseteq \mathbb{F}^n $ such that $ \dim((\mathcal{C}_1 \cap \mathcal{V}_\mathcal{L})/(\mathcal{C}_2 \cap \mathcal{V}_\mathcal{L})) \geq r $ and $ \dim(\mathcal{L}) = d_r $.

There exists a vector space $ \mathcal{D} \subseteq \mathcal{C}_1 \cap \mathcal{V}_\mathcal{L} $ with $ \mathcal{D} \cap \mathcal{C}_2 = \{ 0 \} $ and $ \dim(\mathcal{D}) = r $. We have that $ {\rm RSupp}(\mathcal{D}) \subseteq \mathcal{L} $, since $ \mathcal{D} \subseteq \mathcal{V}_\mathcal{L} $, and hence
$$ d_r = \dim(\mathcal{L}) \geq {\rm wt_R}(\mathcal{D}) \geq d^\prime_r. $$
\end{proof}

Thanks to this characterization, we may connect RGMWs with the rank distance \cite{delsartebilinear}. {\color{black}This will be crucial in the next section, where we will use maximum rank distance codes from \cite{delsartebilinear} to obtain optimal universal secure linear codes for noiseless networks.} Recall the definition of minimum rank distance of a linear coset coding scheme, which is a particular case of \cite[Equation (1)]{similarities}, and which is based on the analogous concept for the Hamming metric given in \cite{duursma}:
\begin{equation}
d_R(\mathcal{C}_1,\mathcal{C}_2) = \min \{ {\rm Rk}(C) \mid C \in \mathcal{C}_1, C \notin \mathcal{C}_2 \}.
\label{rank coset distance}
\end{equation}

The following result follows from the previous theorem and the definitions:

\begin{corollary}[\textbf{Minimum rank distance of linear coset coding schemes}] \label{minimum rank distance}
Given nested linear codes $ \mathcal{C}_2 \subsetneqq \mathcal{C}_1 \subseteq \mathbb{F}^{m \times n} $, it holds that
$$ d_R(\mathcal{C}_1,\mathcal{C}_2) = d_{M,1}(\mathcal{C}_1,\mathcal{C}_2). $$
\end{corollary}

{\color{black}By Theorem \ref{th GMW extend DGW}, the previous corollary coincides with item 1 in \cite[Theorem 30]{ravagnaniweights} when $ \mathcal{C}_2 = \{ 0 \} $ and $ m \neq n $. }

We {\color{black}conclude by showing} the connection between RDRPs and RGMWs:

\begin{proposition} [\textbf{Connection between RDRPs and RGMWs}] \label{connection between RDIP RGMW}
Given nested linear codes $ \mathcal{C}_2 \subsetneqq \mathcal{C}_1 \subseteq \mathbb{F}^{m \times n} $ and $ 1 \leq r \leq \dim(\mathcal{C}_1 / \mathcal{C}_2) $, it holds that
$$ d_{M,r}(\mathcal{C}_1, \mathcal{C}_2) = \min \{ \mu \mid K_{M,\mu}(\mathcal{C}_1, \mathcal{C}_2) \geq r \}. $$
\end{proposition}
\begin{proof}
It is proven as \cite[Proof of Lemma 4]{rgrw}.
\end{proof}

{\color{black}
\subsection{Measuring information leakage on networks} \label{subsec measuring info leakage}
}

In this subsection, {\color{black}we show how the introduced parameters (RGMWs and RDRPs) measure the universal security performance of nested linear code pairs.}

Assume that a given source wants to convey the message $ \mathbf{x} \in \mathbb{F}_q^\ell $, which we assume is a random variable with uniform distribution over $ \mathbb{F}_q^\ell $. Following Subsection \ref{subsection coding schemes}, the source encodes $ \mathbf{x} $ into a matrix $ C \in \mathbb{F}_q^{m \times n} $ using nested linear codes $ \mathcal{C}_2 \subsetneqq \mathcal{C}_1 \subseteq \mathbb{F}_q^{m \times n} $. We also assume that the distributions used in the encoding are all uniform (see Subsection \ref{subsection coding schemes}). 

According to the information leakage model in Subsection \ref{subsection secure communication}, item 2, a wire-tapping adversary obtains $ CB^T \in \mathbb{F}_q^{m \times \mu} $, for some matrix $ B \in \mathbb{F}_q^{\mu \times n} $. 

Recall from \cite{cover} the definition of mutual information of two random variables $ X $ and $ Y $:
\begin{equation}
I(X; Y) = H(Y) - H(Y \mid X),
\label{definition mutual info}
\end{equation}
where $ H(Y) $ denotes the entropy of $ Y $ and $ H(Y \mid X) $ denotes the conditional entropy of $ Y $ given $ X $, and where we take logarithms with base $ q $ (see \cite{cover} for more details).

{\color{black}We will need to use the concept of duality with respect to the Hilbert-Schmidt or trace product. In Appendix \ref{app duality theory}, we collect some basic properties of duality of linear codes. We now give the main definitions:

\begin{definition} [\textbf{Hilbert-Schmidt or trace product}] \label{def trace inner product}
Given matrices $ C, D \in \mathbb{F}^{m \times n} $, we define its Hilbert-Schmidt product, or trace product, as
$$ \langle C, D \rangle = {\rm Trace}(C D^T) $$
$$ = \sum_{i=1}^m \mathbf{c}_i \cdot \mathbf{d}_i = \sum_{i=1}^m \sum_{j=1}^n c_{i,j} d_{i,j} \in \mathbb{F}, $$
where $ \mathbf{c}_i $ and $ \mathbf{d}_i $ are the rows of $ C $ and $ D $, respectively, and where $ c_{i,j} $ and $ d_{i,j} $ are their components, respectively. 

Given a vector space $ \mathcal{C} \subseteq \mathbb{F}^{m \times n} $, we denote by $ \mathcal{C}^\perp $ its dual:
\begin{equation*}
\mathcal{C}^\perp = \{ D \in \mathbb{F}^{m \times n} \mid \langle C, D \rangle = 0, \forall C \in \mathcal{C} \}.
%\label{dual}
\end{equation*}
\end{definition}

We first compute the mutual information of the message and the wire-tapper's observation via rank support spaces:
}

\begin{proposition} \label{information leakage calculation}
Given nested linear codes $ \mathcal{C}_2 \subsetneqq \mathcal{C}_1 \subseteq \mathbb{F}_q^{m \times n} $, a matrix $ B \in \mathbb{F}_q^{\mu \times n} $, and the uniform random variables $ \mathbf{x} $ and $ CB^T $, as in {\color{black}the beginning of this subsection}, it holds that
\begin{equation}
I(\mathbf{x}; CB^T) = \dim(\mathcal{C}_2^\perp \cap \mathcal{V}_\mathcal{L}) - \dim(\mathcal{C}_1^\perp \cap \mathcal{V}_\mathcal{L}),
\label{information leakage equation}
\end{equation}
where $ I(\mathbf{x}; CB^T) $ is as in (\ref{definition mutual info}), and where $ \mathcal{L} = {\rm Row}(B) $. 
\end{proposition}
\begin{proof}
Define the map $ f : \mathbb{F}_q^{m \times n} \longrightarrow \mathbb{F}_q^{m \times \mu} $ given by
$$ f(D) = DB^T, $$
for the matrix $ B \in \mathbb{F}_q^{\mu \times n} $. Observe that $ f $ is a linear map. It follows that
$$ H(CB^T) = H(f(C)) = \log_q ( | f(\mathcal{C}_1) |) = \dim(f(\mathcal{C}_1)) $$
$$ = \dim(\mathcal{C}_1) - \dim(\ker (f) \cap \mathcal{C}_1), $$
where the last equality is the well-known first isomorphism theorem. On the other hand, we may similarly compute the conditional entropy:
$$ H(CB^T \mid \mathbf{x}) = H(f(C) \mid \mathbf{x}) = \log_q ( | f(\mathcal{C}_2) |) = \dim(f(\mathcal{C}_2)) $$
$$ = \dim(\mathcal{C}_2) - \dim(\ker (f) \cap \mathcal{C}_2). $$
However, it holds that $ \ker(f) = \mathcal{V}_{\mathcal{L}^\perp} \subseteq \mathbb{F}_q^{m \times n} $ by Proposition \ref{prop characterization}, since $ B $ is a parity check matrix of $ \mathcal{L}^\perp $. Therefore
$$ I(\mathbf{x}; CB^T) = H(CB^T) - H(CB^T \mid \mathbf{x}) $$
$$ = (\dim(\mathcal{C}_1) - \dim(\mathcal{V}_{\mathcal{L}^\perp} \cap \mathcal{C}_1)) - (\dim(\mathcal{C}_2) - \dim(\mathcal{V}_{\mathcal{L}^\perp} \cap \mathcal{C}_2)). $$
Finally, the result follows by Lemmas \ref{lemma dual rank support} and \ref{forney lemma} in Appendix \ref{app duality theory}.
\end{proof}

The following theorem follows from the previous proposition, Corollary \ref{minimum rank distance} and the definitions:

\begin{theorem}[\textbf{Worst-case information leakage}] \label{maximum secrecy}
Given nested linear codes $ \mathcal{C}_2 \subsetneqq \mathcal{C}_1 \subseteq \mathbb{F}_q^{m \times n} $, and integers $ 0 \leq \mu \leq n $ and $ 1 \leq r \leq \dim(\mathcal{C}_1 / \mathcal{C}_2) $, it holds that
\begin{enumerate}
\item
$ \mu = d_{M,r}(\mathcal{C}_2^\perp, \mathcal{C}_1^\perp) $ is the minimum number of links that an adversary needs to wire-tap in order to obtain at least $ r $ units of information (number of bits multiplied by $ \log_2(q) $) of the sent message.
\item
$ r = K_{M,\mu}(\mathcal{C}_2^\perp, \mathcal{C}_1^\perp) $ is the maximum information (number of bits multiplied by $ \log_2(q) $) about the sent message that can be obtained by wire-tapping at most $ \mu $ links of the network.
\end{enumerate}
In particular, $ t = d_R(\mathcal{C}_2^\perp, \mathcal{C}_1^\perp) - 1 $ is the maximum number of links that an adversary may listen to without obtaining any information about the sent message.
\end{theorem}

{\color{black}
\begin{remark}
Proposition \ref{information leakage calculation} extends \cite[Lemma 7, item 2]{rgrw} from $ \mathbb{F}_{q^m} $-linear codes in $ \mathbb{F}_{q^m}^n $ to $ \mathbb{F}_q $-linear codes in $ \mathbb{F}_q^{m \times n} $ due to Lemma \ref{lemma matrix modules are galois} in Subsection \ref{subsec RGMW extend RGRW}. Furthermore, as we will explain in Theorem \ref{th RGMW extend RGRW}, our Theorem \ref{maximum secrecy} extends in the same sense \cite[Theorem 2]{rgrw} and \cite[Corollary 5]{rgrw}. 
\end{remark}
}

\begin{remark}
In Section \ref{sec relation with others}, we will prove that GMWs coincide with DGWs \cite{ravagnaniweights} when using one code {\color{black}($ \mathcal{C}_1^\perp = \{ 0 \} $ in Theorem \ref{maximum secrecy})} and non-square matrices. Hence the results in this subsection prove that DGWs measure the worst-case information leakage in these cases, which has not been proven in the literature yet.
\end{remark}

{\color{black}
\section{Optimal universal secure linear codes for noiseless networks and any packet length}
}

In this {\color{black}section}, we obtain linear coset coding schemes built from nested linear code pairs $ \mathcal{C} \subsetneqq \mathbb{F}^{m \times n} ${\color{black}, which in this section will refer to those with $ \mathcal{C}_2 = \mathcal{C} $ and $ \mathcal{C}_1 = \mathbb{F}^{m \times n} $,} with optimal universal security {\color{black}performance} in the case of finite fields $ \mathbb{F} = \mathbb{F}_q $ {\color{black}(Theorem \ref{th optimal for noiseless networks})}. Recall from Subsection \ref{subsection coding schemes} that these linear coset coding schemes are suitable for noiseless networks, as noticed in \cite{ozarow} (see also Remark \ref{remark noiseless coset}). 

{\color{black}In this section, we consider perfect universal secrecy (the adversary obtains no information after wire-tapping a given number of links), thus we make use of the theory in last section concerning the first RGMW. In Section \ref{sec singleton}, we will consider bounds on the rest of RGMWs, for general code pairs (suitable for noisy networks), and their achievability.
}

\begin{definition} \label{definition info parameter and privacy}
For a nested linear code pair of the form $ \mathcal{C} \subsetneqq \mathbb{F}_q^{m \times n} $, we define its information parameter as $ \ell = \dim(\mathbb{F}_q^{m \times n} / \mathcal{C}) = \dim(\mathcal{C}^\perp) $, that is the maximum number of $ \log_2(q) $ bits of information that the source can convey, and its security parameter $ t $ as the maximum number of links that an adversary may listen to without obtaining any information about the sent message. 
\end{definition}

Due to Theorem \ref{maximum secrecy}, it holds that $ t = d_R(\mathcal{C}^\perp) - 1 $. We study two problems: 
\begin{enumerate}
\item
Find a nested linear code pair $ \mathcal{C} \subsetneqq \mathbb{F}_q^{m \times n} $ with maximum possible security parameter $ t $ when $ m $, $ n $, $ q $ and the information parameter $ \ell $ are fixed and given.
\item
Find a nested linear code pair $ \mathcal{C} \subsetneqq \mathbb{F}_q^{m \times n} $ with maximum possible information parameter $ \ell $ when $ m $, $ n $, $ q $ and the security parameter $ t $ are fixed and given.
\end{enumerate}

We will deduce bounds on these parameters from the Singleton bound on the dimension of rank-metric codes \cite[Theorem 5.4]{delsartebilinear}:

\begin{lemma} [\textbf{\cite[Theorem 5.4]{delsartebilinear}}]
For a linear code $ \mathcal{C} \subseteq \mathbb{F}_q^{m \times n} $, it holds that
\begin{equation}
\dim(\mathcal{C}) \leq \max \{ m,n \} (\min \{ m,n \} - d_R(\mathcal{C}) + 1). 
\label{singleton bound delsarte}
\end{equation}
\end{lemma}

As usual in the literature, we say that $ \mathcal{C} $ is maximum rank distance (MRD) if equality holds in (\ref{singleton bound delsarte}).

Thanks to Theorem \ref{maximum secrecy} and the previous lemma, we may give upper bounds on the attainable parameters in the previous two problems:

\begin{proposition}
Given a nested linear code pair $ \mathcal{C} \subsetneqq \mathbb{F}_q^{m \times n} $ with information parameter $ \ell $ and security parameter $ t $, it holds that:
\begin{equation}
\ell \leq \max \{ m,n \} (\min \{ m,n \} - t),
\end{equation}
\begin{equation}
t \leq \min \{ m,n \} - \left\lceil \frac{\ell}{\max \{ m,n \}} \right\rceil.
\end{equation}
In particular, $ \ell \leq mn $ and $ t \leq \min \{ m,n \} $.
\end{proposition}
\begin{proof}
Recall that $ \ell = \dim(\mathbb{F}_q^{m \times n} / \mathcal{C}) = \dim(\mathcal{C}^\perp) $ and, due to Theorem \ref{maximum secrecy}, $ t = d_R(\mathcal{C}^\perp) - 1 $. Hence the result follows from the bound (\ref{singleton bound delsarte}) for $ \mathcal{C}^\perp $.
\end{proof}

On the other hand, the existence of linear codes in $ \mathbb{F}_q^{m \times n} $ attaining the Singleton bound on their dimensions, for all possible choices of $ m $, $ n $ and minimum rank distance $ d_R $ \cite[Theorem 6.3]{delsartebilinear}, leads to the following existence result on optimal linear coset coding schemes for noiseless networks.

\begin{theorem} \label{th optimal for noiseless networks}
For all choices of positive integers $ m $ and $ n $, and all finite fields $ \mathbb{F}_q $, the following hold:
\begin{enumerate}
\item
For every positive integer $ \ell \leq mn $, there exists a nested linear code pair $ \mathcal{C} \subsetneqq \mathbb{F}_q^{m \times n} $ with information parameter $ \ell $ and security parameter $ t = \min \{ m,n \} - \left\lceil (\ell / \max \{ m,n \} ) \right\rceil $.
\item
For every positive integer $ t \leq \min \{ m,n \} $, there exists a nested linear code pair $ \mathcal{C} \subsetneqq \mathbb{F}_q^{m \times n} $ with security parameter $ t $ and information parameter $ \ell = \max \{ m,n \} (\min \{ m,n \} - t) $.
\end{enumerate}
\end{theorem}

\begin{remark} \label{remark other optimal on noiseless}
We remark here that, to the best of our knowledge, only the linear coset coding schemes in item 2 in the previous theorem, for the special case $ n \leq m $, have been obtained in the literature. It corresponds to \cite[Theorem 7]{silva-universal}.

Using cartesian products of MRD codes as in \cite[Subsection VII-C]{silva-universal}, linear coset coding schemes as in item 2 in the previous theorem can be obtained when $ n > m $, for the restricted parameters $ n = lm $ and $ \ell = m l k^\prime $, where $ l $ and $ k^\prime < m $ are positive integers.
\end{remark}

\section{Universal secure list-decodable rank-metric linear coset coding schemes} \label{sec list decoding}

In this section, we will {\color{black}obtain} nested linear code pairs $ \mathcal{C}_2 \subsetneqq \mathcal{C}_1 \subseteq \mathbb{F}_q^{m \times n} $ {\color{black}when $ n $ divides $ m $} that can list-decode rank errors {\color{black}on noisy networks (as opposed to the scenario in last section),} whose list sizes are polynomial on the code length $ n $, while being univeral secure under a given number of wire-tapped links. {\color{black}As in last section, we consider perfect universal secrecy, and thus make use of the results in Section \ref{sec introduction of RGMWs} concerning the first RGMW of the dual code pair.}

{\color{black}We give the construction in Subsection \ref{subsec list decodable construction}, together with its parameters (Theorem \ref{th secure list-decodable}): information parameter $ \ell $, security parameter $ t $ and number of list-decodable rank errors $ e $. To measure the quality of the proposed code pair, we will compare in Subsection \ref{subsec list decoding comparison gabidulin} their} parameters with those obtained when choosing $ \mathcal{C}_1 $ and $ \mathcal{C}_2 $ as MRD codes \cite{gabidulin, roth}{\color{black}, which provide coset coding schemes with both optimal universal security and optimal error-correction capability \cite{silva-universal}}. We will also show (Subsection \ref{subsec near optimality list decoding}) the near optimality of the obtained construction in terms of the introduced uncertainty on the secret message and the number of list-decodable rank errors.

{\color{black}
\subsection{The construction and its main properties} \label{subsec list decodable construction}

We start by extending the definition of rank list-decodable codes from \cite[Definition 2]{ding} to coset coding schemes:

\begin{definition}
For positive integers $ e $ and $ L $, we say that a coset coding scheme $ \mathcal{P}_\mathcal{S} = \{ \mathcal{C}_\mathbf{x} \}_{\mathbf{x} \in \mathcal{S}} $ over $ \mathbb{F}_q $ is rank $ (e, L) $-list-decodable if, for every $ Y \in \mathbb{F}_q^{m \times n} $, we have that
$$ | \left\lbrace \mathbf{x} \in \mathcal{S} \mid \mathcal{P}_\mathbf{x} \cap \mathcal{B}(Y,e) \neq \emptyset \right\rbrace | \leq L, $$
where $ \mathcal{B}(Y,e) $ denotes the ball in $ \mathbb{F}_q^{m \times n} $ with center $ Y $ and rank radius $ e $. The number of list-decodable rank errors is $ e $ and the list sizes are said to be polynomial in $ n $ if $ L = \mathcal{O}(F(n)) $, for some polynomial $ F(x) $.
\end{definition}

\begin{remark}
Observe however that, if a coset coding scheme can list-decode $ e $ rank errors with polynomial-sized lists of cosets, we still need to decode these cosets to obtain the uncoded secret messages. In general, it is possible that the union of such cosets has exponential size while the scheme can still obtain all the corresponding uncoded messages via an algorithm with polynomial complexity. This is the case in the construction below.
\end{remark}

We now give the above mentioned construction, which exists whenever $ n $ divides $ m $. The main objective is to obtain simultaneously large information parameter $ \ell $, security parameter $ t $ and number of list-decodable rank errors $ e $.}

\begin{construction} \label{construction list decodable}
{\color{black}Assume that $ n $ divides $ m $ and fix $ \varepsilon > 0 $ and positive integers $ s $ and $ 1 \leq k_2 < k_1 \leq n $ such that $ 4 s n \leq \varepsilon m $ and $ m/n = \mathcal{O}(s / \varepsilon) $. In the next subsection, $ mk_1 $ and $ mk_2 $ will be the dimensions of the MRD linear codes constituting an optimal universal secure nested coset coding scheme, but here they are just fixed parameters. }

Fix a basis $ \alpha_1, \alpha_2, \ldots, \alpha_m $ of $ \mathbb{F}_{q^m} $ as a vector space over $ \mathbb{F}_q $, such that $ \alpha_1, \alpha_2, \ldots, \alpha_n $ generate $ \mathbb{F}_{q^n} $ (recall that $ \mathbb{F}_{q^n} \subseteq \mathbb{F}_{q^m} $ since $ n $ divides $ m $).  

Recall that a $ q $-linearized polynomial over $ \mathbb{F}_{q^m} $ is a polynomial of the form $ F(x) = \sum_{i=0}^d F_i x^{q^i} $, where $ F_i \in \mathbb{F}_{q^m} $, for some positive integer $ d $. Denote also $ {\rm ev}_{\boldsymbol\alpha}(F(x)) = (F(\alpha_1), $ $ F(\alpha_2), \ldots, F(\alpha_n)) \in \mathbb{F}_{q^m}^n $, and finally define {\color{black}the linear codes}
$$ \mathcal{C}_2 = \{ M_{\boldsymbol\alpha}({\rm ev}_{\boldsymbol\alpha}(F(x))) \mid F_i = 0 \textrm{ for } i < k_1-k_2 \textrm{ and } i \geq k_1 \}, $$
$$ \mathcal{C}_1 = \{ M_{\boldsymbol\alpha}({\rm ev}_{\boldsymbol\alpha}(F(x))) \mid F_i \in \mathcal{H}_i \textrm{ for } 0 \leq i < k_1-k_2, $$
$$ F_i \in \mathbb{F}_{q^m} \textrm{ for } k_1-k_2 \leq i < k_1, F_i = 0 \textrm{ for } i \geq k_1 \}, $$
where $ M_{\boldsymbol\alpha} $ is the map given in (\ref{equation matrix representation}) and $ \mathcal{H}_0, \mathcal{H}_1, $ $ \ldots, $ $ \mathcal{H}_{k_1 - k_2 - 1} \subseteq \mathbb{F}_{q^m} $ are the $ \mathbb{F}_q $-linear vector spaces described in \cite[Theorem 8]{list-decodable-rank-metric}. {\color{black}We recall this description in Appendix \ref{app explicit designs}. Observe that these vector spaces depend on $ \varepsilon $ and $ s $.} 

{\color{black}Let $ \ell = \dim(\mathcal{C}_1 / \mathcal{C}_2) = \dim (\mathcal{H}_0 \times \mathcal{H}_1 \times \cdots \times \mathcal{H}_{k_1 - k_2 - 1}) $. We now show how $ \mathcal{C}_2 \subsetneqq \mathcal{C}_1 \subseteq \mathbb{F}_q^{m \times n} $ form a coset coding scheme as in Definition \ref{definition NLCP}. Define the vector space
$$ \mathcal{W} = \{ M_{\boldsymbol\alpha}({\rm ev}_{\boldsymbol\alpha}(F(x))) \mid F_i \in \mathcal{H}_i \textrm{ for } i < k_1-k_2 $$
$$ \textrm{ and } F_i = 0 \textrm{ for } i \geq k_1-k_2 \}, $$
which satisfies that $ \mathcal{C}_1 = \mathcal{C}_2 \oplus \mathcal{W} $. Now consider the secret space as $  \mathcal{H}_0 \times \mathcal{H}_1 \times \cdots \times \mathcal{H}_{k_1 - k_2 - 1} \cong \mathbb{F}_q^\ell $, and define the vector space isomorphism $ \psi : \mathcal{H}_0 \times \mathcal{H}_1 \times \cdots \times \mathcal{H}_{k_1 - k_2 - 1} \longrightarrow \mathcal{W} $ as follows: For $ \mathbf{x} \in \mathcal{H}_0 \times \mathcal{H}_1 \times \cdots \times \mathcal{H}_{k_1 - k_2 - 1} $, take $ F(x) = \sum_{i=0}^{k_1 - k_2 - 1} F_i x^{q^i} $ such that $ \mathbf{x} = (F_0, F_1, \ldots, F_{k_1 - k_2 - 1}) $, and define 
$$ C = \psi(\mathbf{x}) = M_{\boldsymbol\alpha}({\rm ev}_{\boldsymbol\alpha}(F(x))). $$
}
\end{construction}

{\color{black}We may now state the main result of this section:

\begin{theorem} \label{th secure list-decodable}
With the same assumptions and notation, the nested coset coding scheme in Construction \ref{construction list decodable} satisfies that:
\begin{enumerate}
\item
$ \ell = \dim(\mathcal{C}_1 / \mathcal{C}_2) \geq m(k_1 - k_2)(1 - 2 \varepsilon) $.
\item
Its security parameter (Definition \ref{definition info parameter and privacy}) satisfies $ t \geq k_2 $.
\item
It is rank $ (e,L) $-list-decodable for all $ e \leq \frac{s}{s+1}(n-k_1) $, with $ L \leq q^{\mathcal{O}(s^2/\varepsilon^2)} $, and it admits a list-decoding algorithm that obtains all corresponding uncoded messages with polynomial complexity in $ n $.
\end{enumerate}
\end{theorem}

We devote the rest of the subsection to prove this theorem. We need to recall some definitions and results from \cite{list-decodable-rank-metric}:

\begin{definition}[\textbf{Subspace designs \cite[Definition 3]{list-decodable-rank-metric}}]
Assuming that $ n $ divides $ m $ and given positive integers $ r $ and $ N $, a collection of $ \mathbb{F}_q $-linear subspaces $ \mathcal{U}_1, \mathcal{U}_2, \ldots, \mathcal{U}_M \subseteq \mathbb{F}_{q^m} $ is called an $ (r,N,n) $ $ \mathbb{F}_q $-linear subspace design if
$$ \sum_{i=1}^M \dim(\mathcal{U}_i \cap \mathcal{V}) \leq N, $$
with dimensions taken over $ \mathbb{F}_q $, for every $ \mathbb{F}_{q^n} $-linear subspace $ \mathcal{V} \subseteq \mathbb{F}_{q^m} $ of dimension at most $ r $ over $ \mathbb{F}_{q^n} $.
\end{definition}

The following lemma is part of \cite[Theorem 8]{list-decodable-rank-metric}:

\begin{lemma}[\textbf{\cite{list-decodable-rank-metric}}] \label{lemma the Hi are subspace design}
With assumptions and notation as in Construction \ref{construction list decodable}, the spaces $ \mathcal{H}_0, \mathcal{H}_1, \ldots \mathcal{H}_{k_1 - k_2 - 1} $ defined in Appendix \ref{app explicit designs} form an $ (s, 2(m/n - 1)s/\varepsilon, n) $ $ \mathbb{F}_q $-linear subspace design.
\end{lemma}

\begin{definition}[\textbf{Periodic subspaces \cite[Definition 9]{list-decodable-rank-metric}}]
Given positive integers $ r, l, k $, we say that an affine subspace $ \mathcal{H} \subseteq \mathbb{F}_{q^n}^{lk} $ is $ (r,l,k) $-periodic if there exists an $ \mathbb{F}_{q^n} $-linear subspace $ \mathcal{V} \subseteq \mathbb{F}_{q^n}^l $ of dimension at most $ r $ over $ \mathbb{F}_{q^n} $ such that, for every $ j = 2,3, \ldots, k $ and $ \mathbf{a} \in \mathbb{F}_{q^n}^{(j-1)l} $, the affine space
$$ \{ \pi_{[(j-1)l+1,jl]}(\mathbf{x}) \mid \mathbf{x} \in \mathcal{H}, \pi_{[1,(j-1)l]}(\mathbf{x}) = \mathbf{a} \} \subseteq \mathbb{F}_{q^n}^l $$
is contained in $ \mathbf{v}_\mathbf{a} + \mathcal{V} $, for a vector $ \mathbf{v}_\mathbf{a} \in \mathbb{F}_{q^n}^l $ that depends on $ \mathbf{a} $. Here, $ \pi_J $ denotes the projection over the coordinates in $ J $, and $ [a,b] $ denotes the set of integers $ i $ such that $ a \leq i \leq b $.
\end{definition}

We may now prove our main result:

\begin{proof}[Proof of Theorem \ref{th secure list-decodable}]
We prove each item separately:

1) By Lemma \ref{lemma dimensions of Hi} in Appendix \ref{app explicit designs}, it holds that $ \dim(\mathcal{H}_i) \geq m (1 - 2 \varepsilon) $, for $ i = 0,1,2, \ldots, k_1 - k_2 - 1 $. Therefore
$$ \ell = \dim(\mathcal{H}_0 \times \mathcal{H}_1 \times \cdots \times \mathcal{H}_{k_1 - k_2 - 1}) \geq m (k_1 - k_2) (1 - 2 \varepsilon). $$

2) By Theorem \ref{maximum secrecy}, the security parameter is $ t = d_R(\mathcal{C}_2^\perp, \mathcal{C}_1^\perp) - 1 \geq d_R(\mathcal{C}_2^\perp) -1 $. Since $ \mathcal{C}_2 $ is MRD, then so is its trace dual \cite{delsartebilinear}, which means that $ d_R(\mathcal{C}_2^\perp) = k_2 + 1 $, and the result follows.

3) As shown in \cite[Subsection IV-B]{list-decodable-rank-metric}, we may perform list-decoding for the Gabidulin code $ \mathcal{G}_1 \supseteq \mathcal{C}_1 $,
$$ \mathcal{G}_1 = \{ M_{\boldsymbol\alpha}({\rm ev}_{\boldsymbol\alpha}(F(x))) \mid F_i = 0 \textrm{ for } i \geq k_1 \}, $$
and obtain in polynomial time a list containing all possible sent messages that is an $ (s-1, m/n, k_1) $-periodic subspace of $ \mathbb{F}_{q^n}^{k_1 m/n} \cong \mathbb{F}_{q^m}^{k_1} $ (isomorphic as $ \mathbb{F}_{q^n} $-linear vector spaces).

Project this periodic subspace onto the first $ k_1 - k_2 $ coordinates, which gives a $ (s-1, m/n, k_1-k_2) $-periodic subspace of $ \mathbb{F}_{q^m}^{k_1 - k_2} $, and intersect it with $ \mathcal{H}_0 \times \mathcal{H}_1 \times \cdots \times \mathcal{H}_{k_1 - k_2 - 1} $. Since $ \mathcal{H}_0, \mathcal{H}_1, \ldots \mathcal{H}_{k_1 - k_2 - 1} $ form an $ (s, 2(m/n - 1)s/\varepsilon, n) $ $ \mathbb{F}_q $-linear subspace design by Lemma \ref{lemma the Hi are subspace design}, such intersection is an $ \mathbb{F}_q $-linear affine space of dimension at most $ \mathcal{O}(s^2 / \varepsilon^2) $ (recall that $ m/n = \mathcal{O}(s / \varepsilon) $) by the definition of subspace designs and periodic subspaces. 
\end{proof}
}

{\color{black}
\subsection{Comparison with optimal unique-decodable linear coset coding schemes based on MRD codes} \label{subsec list decoding comparison gabidulin}

In this subsection, we compare the schemes in Construction \ref{construction list decodable} with those obtained when using MRD codes \cite{gabidulin, roth}, whose information parameter $ \ell $ is optimal for given security parameter $ t $ and number of unique-decodable rank errors $ e $, due to Theorems 11 and 12 in \cite{silva-universal}.

\begin{proposition} [\textbf{\cite{silva-universal}}]
Assume that $ n \leq m $ and $ \mathcal{C}_2 \subsetneqq \mathcal{C}_1 \subseteq \mathbb{F}_q^{m \times n} $ are MRD linear codes of dimensions $ \dim(\mathcal{C}_1) = m k_1 $ and $ \dim(\mathcal{C}_2) = m k_2 $ (recall that, by the Singleton bound (\ref{singleton bound delsarte}), dimensions of MRD codes are multiple of $ m $ when $ n \leq m $).

The linear coset coding scheme {\color{black}(Definition \ref{definition NLCP})} constructed from this nested linear code pair satisfies that:
\begin{enumerate}
\item
Its information parameter is $ \ell = m(k_1 - k_2) $.
\item
Its security parameter is $ t = k_2 $.
\item
If the number of rank errors is $ e \leq \lfloor \frac{n - k_1}{2} \rfloor $, then rank error-correction can be performed, giving a unique solution.
\end{enumerate}
\end{proposition}

Therefore, assuming that $ n $ divides $ m $ and given MRD linear codes $ \mathcal{C}_2 \subsetneqq \mathcal{C}_1 \subseteq \mathbb{F}_q^{m \times n} $ of dimensions $ \dim(\mathcal{C}_1) = m k_1 $ and $ \dim(\mathcal{C}_2) = m k_2 $, the linear coset coding scheme in Construction \ref{construction list decodable} has at least the same security parameter $ t $ as that obtained using $ \mathcal{C}_1 $ and $ \mathcal{C}_2 $, an information parameter $ \ell $ that is at least $ 1 - 2 \varepsilon $ times the one obtained using $ \mathcal{C}_1 $ and $ \mathcal{C}_2 $, and can list-decode in polynomial time (with list of polynomial size) roughly $ n - k_1 $ errors, which is twice as many as the rank errors that $ \mathcal{C}_1 $ and $ \mathcal{C}_2 $ can correct, due to the previous proposition and Theorem \ref{th secure list-decodable}. }

\subsection{Near optimality of the obtained construction} \label{subsec near optimality list decoding}

{\color{black}In this subsection,} we will show the near optimality of {\color{black}Construction \ref{construction list decodable} in terms of its \textit{introduced uncertainty} $ H(C|\mathbf{x}) $ compared to the maximum \textit{observed information} $ H(CB^T) $ by the wire-tapper,} and the number of rank errors $ e $ that the scheme can list-decode.

Let $ \mathbf{x} \in \mathbb{F}_q^\ell $ and $ C \in \mathbb{F}_q^{m \times n} $ denote the random variables representing the secret message and the transmitted codeword, respectively, as in Subsection \ref{subsec measuring info leakage}. 

{\color{black}The quantity $ H(C|\mathbf{x}) $ measures the amount of randomness of $ C $ given $ \mathbf{x} $ introduced by the corresponding coset coding scheme, and we would like it to be as small as possible since generating randomness is difficult in practice. Observe that $ H(C|\mathbf{x}) = \dim(\mathcal{C}_2) $ for nested coset coding schemes. On the other hand, the quantity $ H(CB^T) $ measures the amount of observed information by wire-tapping $ \mu $ links if $ B \in \mathbb{F}_q^{\mu \times n} $, which satisfies $ H(CB^T) \leq m \mu $, being the inequality usually tight when $ I(\mathbf{x};CB^T) = 0 $ or even an equality, as is the case for Gabidulin codes. Thus the following bound is a weaker version of a bound of the form $ m t \leq \dim(\mathcal{C}_2) $, which we leave as open problem.

\begin{proposition}
Fix an arbitrary coset coding scheme in $ \mathbb{F}_q^{m \times n} $ with message set $ \mathcal{S} = \mathbb{F}_q^\ell $, let $ \mathbf{x} \in \mathbb{F}_q^\ell $, and let $ C \in \mathbb{F}_q^{m \times n} $ be its encoding. It holds that
$$ \max \{ H(CB^T) \mid B \in \mathbb{F}_q^{\mu \times n}, I(\mathbf{x};CB^T) = 0 \} \leq H(C|\mathbf{x}). $$
\end{proposition}
\begin{proof}
Fix $ B \in \mathbb{F}_q^{\mu \times n} $. The result follows from the following chain of inequalities:
\begin{eqnarray*}
&&I(\mathbf{x};CB^T) \\
&=& H(CB^T) - H(CB^T|\mathbf{x}) \\
&=& H(CB^T) - H(CB^T|C,\mathbf{x}) \\
&& + H(CB^T|C,\mathbf{x}) - H(CB^T|\mathbf{x}) \\
&=& H(CB^T) - H(CB^T|C) \\
&& + H(CB^T|C,\mathbf{x}) - H(CB^T|\mathbf{x})\\
&&\mbox{(since } \mathbf{x} \rightarrow C \rightarrow CB^T
\mbox{ is a Markov chain \cite{cover})}\\
&=& I(C;CB^T) - I(C;CB^T|\mathbf{x}) \\
&\geq& H(CB^T) - H(C|\mathbf{x}).
\end{eqnarray*}
\end{proof}

Now consider the coset coding scheme in Construction \ref{construction list decodable}, and fix} $ \mu \leq k_2 \leq k_1 $. Define the Gabidulin code
$$ \mathcal{G}_1 = \{ M_{\boldsymbol\alpha}({\rm ev}_{\boldsymbol\alpha}(F(x))) \mid F_i = 0, i \geq k_1 \} \subseteq \mathbb{F}_q^{m \times n}, $$
and let $ G $ be the uniform random variable on $ \mathcal{G}_1 $. It holds that
\begin{equation}
  \max_{B\in \mathbb{F}_q^{\mu \times n}}
  H(GB^T) = m\mu, \label{eq3}
\end{equation}
since $ \mu \leq k_1 $. Equation (\ref{eq3}) together with $ \dim(\mathcal{G}_1/\mathcal{C}_1) \leq 2m \varepsilon (k_1-k_2) $ implies that
$$ \max_{B\in \mathbb{F}_q^{\mu \times n}} H(CB^T) \geq m(\mu -2\varepsilon(k_1-k_2)). $$
Using that $ H(C|\mathbf{x}) = \dim(\mathcal{C}_2) = mk_2 $, we see that {\color{black}the bound in the previous proposition is tight for Construction \ref{construction list decodable}:}
\begin{equation*}
\begin{split}
0 \leq & H(C|\mathbf{x}) - \max \{ H(CB^T) \mid B \in \mathbb{F}_q^{\mu \times n}, I(\mathbf{x};CB^T) = 0 \} \\
\leq & m (k_2 - t + 2 \varepsilon (k_1 - k_2)) \leq 2 \varepsilon m (k_1 - k_2). 
\end{split}
\end{equation*}

Next we show that the rank list-decoding capability cannot be improved for large $ s $ {\color{black}and small $ \varepsilon $, compared to general nested coset coding schemes. Since rank list-decodable nested coset coding schemes still require decoding each coset, we will consider those such that a complementary space $ \mathcal{W} $ as in Definition \ref{definition NLCP} is rank list-decodable with polynomial-sized lists after adding an error matrix from the smaller code $ \mathcal{C}_2 $:

\begin{proposition}
Fix a nested linear code pair $ \mathcal{C}_2 \varsubsetneqq \mathcal{C}_1 \subseteq \mathbb{F}_q^{m \times n} $ and a subspace $ \mathcal{W} \subseteq \mathcal{C}_1 $ such that $ \mathcal{C}_1 = \mathcal{C}_2 \oplus \mathcal{W} $, and denote by $ M $ the maximum rank of a matrix in $ \mathcal{C}_2 $. If $ \mathcal{W} $ is rank $ (e + M, L) $-list-decodable with polynomial list sizes $ L $, then
$$ e \leq n - \frac{\dim(\mathcal{C}_1)}{m}. $$
\end{proposition}
\begin{proof}
By \cite[Proposition 1]{ding}, if the linear code $ \mathcal{W} $ is rank $ (e+M, L) $-list-decodable with polynomial-sized lists $ L $, then 
$$ e + M \leq n - \dim(\mathcal{W})/m. $$
On the other hand, the maximum rank of codewords in $ \mathcal{C}_2 $ is at least $ \dim(\mathcal{C}_2) / m $ by \cite[Proposition 47]{ravagnani}. Hence
$$ e \leq n - \frac{\dim(\mathcal{W})}{m} - \frac{\dim(\mathcal{C}_2)}{m} = n - \frac{\dim(\mathcal{C}_1)}{m}, $$
and we are done.
\end{proof}

For the nested coset coding scheme in Construction \ref{construction list decodable}, it holds that
$$ e = \frac{s}{s+1}(n-k_1), \textrm{ and} $$
$$ n - \frac{\dim(\mathcal{C}_1)}{m} = n - k_1(1- 2 \varepsilon) - 2 \varepsilon k_2, $$
which are closer as $ s $ becomes larger and $ \varepsilon $ becomes smaller.
}

\section{Security equivalences of linear coset coding schemes and minimum parameters}

In this section, we study when two nested linear code pairs $ \mathcal{C}_2 \subsetneqq \mathcal{C}_1 \subseteq \mathbb{F}^{m \times n} $ and $ \mathcal{C}_2^\prime \subsetneqq \mathcal{C}_1^\prime \subseteq \mathbb{F}^{m^\prime \times n^\prime} $ have the same {\color{black}universal} security and/or reliability performance. 

{\color{black}First, we define security equivalences and give several characterizations of these in Theorem \ref{theorem characterization} (Subsection \ref{subsec characterizations equivalences}), which show that they also preserve error and erasure correction capabilities. As applications, we study ranges and minimum possible parameters $ m $ and $ n $ for linear codes (Subsection \ref{subsec minimum parameters}), and we study when they are degenerate (Subsection \ref{subsec degenerate}), meaning when they can be applied to networks with strictly smaller length $ n $. }

\subsection{Security equivalences and rank isometries} \label{subsec characterizations equivalences}

{\color{black}In this subsection, we first give in Theorem \ref{theorem characterization} the above mentioned characterizations, and we define afterwards security equivalences as maps satisfying one of such characterizations. We continue with Proposition \ref{prop motivating security equivalence}, which shows that security equivalences actually preserve universal security performance as in Subsection \ref{subsection secure communication}, thus motivating our definition. We conclude by comparing Theorem \ref{theorem characterization} with related results from the literature (see also Table \ref{table isomorphism charact}).

Due to the importance of the rank metric for error and erasure correction in linear network coding (see Subsection \ref{subsection secure communication}), and for universal security (by Theorem \ref{maximum secrecy} and Corollary \ref{minimum rank distance}), we start by considering rank isometries:}

\begin{definition} [\textbf{Rank isometries}]
We say that a map $ \phi : \mathcal{V} \longrightarrow \mathcal{W} $ between vector spaces $ \mathcal{V} \subseteq \mathbb{F}^{m \times n} $ and $ \mathcal{W} \subseteq \mathbb{F}^{m^\prime \times n^\prime} $ is a rank isometry if it is a vector space isomorphism and $ {\rm Rk}(\phi(V)) = {\rm Rk}(V) $, for all $ V \in \mathcal{V} $. In that case, we say that $ \mathcal{V} $ and $ \mathcal{W} $ are rank isometric.
\end{definition}

We have the following result, which was first proven in \cite[Theorem 1]{marcus} for square matrices and the complex field $ \mathbb{F} = \mathbb{C} $. In \cite[Proposition 3]{morrison} it is observed that the square condition is not necessary and it may be proven for arbitrary fields:
 
\begin{proposition} [\textbf{\cite{marcus, morrison}}] \label{morrison proposition}
If $ \phi : \mathbb{F}^{m \times n} \longrightarrow \mathbb{F}^{m \times n} $ is a rank isometry, then there exist invertible matrices $ A \in \mathbb{F}^{m \times m} $ and $ B \in \mathbb{F}^{n \times n} $ such that 
\begin{enumerate}
\item
$ \phi(C) = ACB $, for all $ C \in \mathbb{F}^{m \times n} $, or
\item
$ \phi(C) = AC^TB $, for all $ C \in \mathbb{F}^{m \times n} $,
\end{enumerate}
where the latter case can only happen if $ m=n $.
\end{proposition}

{\color{black}We will define security equivalences as certain vector space isomorphisms satisfying one of several equivalent conditions. We first show their equivalence in the following theorem, which is the main result of this section: }

\begin{theorem} \label{theorem characterization}
Let $ \phi : \mathcal{V} \longrightarrow \mathcal{W} $ be a vector space isomorphism between rank support spaces $ \mathcal{V} \in RS(\mathbb{F}^{m \times n}) $ and $ \mathcal{W} \in RS(\mathbb{F}^{m \times n^\prime}) $, and consider the following properties:
\begin{itemize}
\item[(P 1)] There exist full-rank matrices $ A \in \mathbb{F}^{m \times m} $ and $ B \in \mathbb{F}^{n \times n^\prime} $ such that $ \phi(C) = ACB $, for all $ C \in\mathcal{V} $.
\item[(P 2)] A subspace $ \mathcal{U} \subseteq \mathcal{V} $ is a rank support space if, and only if, $ \phi(\mathcal{U}) $ is a rank support space. 
\item[(P 3)] For all subspaces $ \mathcal{D} \subseteq \mathcal{V} $, it holds that $ {\rm wt_R}(\phi(\mathcal{D})) = {\rm wt_R}(\mathcal{D}) $.
\item[(P 4)] $ \phi $ is a rank isometry.
\end{itemize}
Then the following implications hold:
$$ (\textrm{P 1}) \Longleftrightarrow (\textrm{P 2}) \Longleftrightarrow (\textrm{P 3}) \Longrightarrow (\textrm{P 4}). $$
In particular, a security equivalence is a rank isometry and, in the case $ \mathcal{V} = \mathcal{W} = \mathbb{F}^{m \times n} $ and $ m \neq n $, the reversed implication holds by Proposition \ref{morrison proposition}.
\end{theorem}
\begin{proof}
See Appendix \ref{app 1}.
\end{proof}

\begin{remark}
Unfortunately, the implication $ (\textrm{P 3}) \Longleftarrow (\textrm{P 4}) $ does not always hold. Take for instance $ m=n $ and the map $ \phi : \mathbb{F}^{m \times m} \longrightarrow \mathbb{F}^{m \times m} $ given by $ \phi(C) = C^T $, for all $ C \in \mathbb{F}^{m \times m} $. 
\end{remark}

\begin{remark} \label{remark sec equivalences preserve RGMWs}
Observe that{\color{black}, in particular,} security equivalences {\color{black}also} preserve (relative) generalized matrix weights{\color{black}, (relative) dimension/rank support profiles} and distributions of rank weights of vector subspaces, and they are the only rank isometries with these properties.
\end{remark}

{\color{black}Property (P 1) will be useful for technical computations and, in particular, for Proposition \ref{prop motivating security equivalence} below. As explained in Appendix \ref{app 1}, (P 2) allows us to connect (P 1) with (P 3), and (P 3) allows us to connect the first two with the rank metric (P 4), crucial for error and erasure correction as in Subsection \ref{subsection secure communication}. Finally, Property (P 2) also explains why we will consider security equivalences defined between rank support spaces, and intuitively explains that such spaces behave as ambient spaces in our theory, as mentioned in Subsection \ref{subsec rank supports}.  }

\begin{definition} [\textbf{Security equivalences}] \label{def security equivalences}
We say that a map $ \phi : \mathcal{V} \longrightarrow \mathcal{W} $ between rank support spaces $ \mathcal{V} \in RS(\mathbb{F}^{m \times n}) $ and $ \mathcal{W} \in RS(\mathbb{F}^{m \times n^\prime}) $ is a security equivalence if it is a vector space isomorphism and {\color{black}satisfies condition (P 1), (P 2) or (P 3) in Theorem \ref{theorem characterization}.}

Two nested linear code pairs $ \mathcal{C}_2 \subsetneqq \mathcal{C}_1 \subseteq \mathbb{F}^{m \times n} $ and $ \mathcal{C}^\prime_2 \subsetneqq \mathcal{C}^\prime_1 \subseteq \mathbb{F}^{m \times n^\prime} $ are said to be security equivalent if there exist rank support spaces $ \mathcal{V} \in RS(\mathbb{F}^{m \times n}) $ and $ \mathcal{W} \in RS(\mathbb{F}^{m \times n^\prime}) $, containing $ \mathcal{C}_1 $ and $ \mathcal{C}^\prime_1 $, respectively, and a security equivalence $ \phi : \mathcal{V} \longrightarrow \mathcal{W} $ with $ \phi(\mathcal{C}_1) = \mathcal{C}^\prime_1 $ and $ \phi(\mathcal{C}_2) = \mathcal{C}^\prime_2 $.
\end{definition}

{\color{black}We now motivate the previous definition with the next proposition, which makes use of Theorem \ref{theorem characterization}. Observe that Remark \ref{remark sec equivalences preserve RGMWs} above already shows that security equivalences preserve the worst-case information leakage as described in Theorem \ref{maximum secrecy}. Now, given nested linear code pairs $ \mathcal{C}_2 \subsetneqq \mathcal{C}_1 \subseteq \mathbb{F}_q^{m \times n} $ and $ \mathcal{C}^\prime_2 \subsetneqq \mathcal{C}^\prime_1 \subseteq \mathbb{F}_q^{m \times n^\prime} $, Proposition \ref{prop motivating security equivalence} below shows that if the dual pairs are security equivalent, then there exists a bijective correspondence between wire-tappers' transfer matrices (matrix $ B $ in Subsection \ref{subsection secure communication}, item 2) that preserves the mutual information with the original sent message. If the original pairs are also security equivalent, we conclude that encoding with $ \mathcal{C}_2 \subsetneqq \mathcal{C}_1 \subseteq \mathbb{F}_q^{m \times n} $ or $ \mathcal{C}^\prime_2 \subsetneqq \mathcal{C}^\prime_1 \subseteq \mathbb{F}_q^{m \times n^\prime} $ yields exactly the same universal error and erasure correction performance, and exactly the same universal security performance over linearly coded networks, as in Subsection \ref{subsection secure communication}.

\begin{proposition} \label{prop motivating security equivalence}
Assume that $ \mathbb{F} = \mathbb{F}_q $ and the dual pairs of $ \mathcal{C}_2 \subsetneqq \mathcal{C}_1 \subseteq \mathbb{F}_q^{m \times n} $ and $ \mathcal{C}^\prime_2 \subsetneqq \mathcal{C}^\prime_1 \subseteq \mathbb{F}_q^{m \times n^\prime} $ are security equivalent by a security equivalence given by matrices $ A \in \mathbb{F}_q^{m \times m} $ and $ B \in \mathbb{F}_q^{n \times n^\prime} $ as in item 1 in Theorem \ref{theorem characterization}. For any matrix $ M \in \mathbb{F}_q^{\mu \times n} $, it holds that
\begin{equation}
I \left( \mathbf{x} ; CM^T \right) = I \left( \mathbf{x} ; C^\prime (M B)^T \right),
\label{eq security equivalences preserve info}
\end{equation}
with notation as in Proposition \ref{information leakage calculation}, where $ C \in \mathbb{F}_q^{m \times n} $ and $ C^\prime \in \mathbb{F}_q^{m \times n^\prime} $ are the encodings of $ \mathbf{x} $ using $ \mathcal{C}_2 \subsetneqq \mathcal{C}_1 \subseteq \mathbb{F}_q^{m \times n} $ and $ \mathcal{C}^\prime_2 \subsetneqq \mathcal{C}^\prime_1 \subseteq \mathbb{F}^{m \times n^\prime} $, respectively.

Furthermore, assuming $ n \leq n^\prime $, the correspondence $ M \mapsto MB $ is one to one and, for any matrix $ N \in \mathbb{F}_q^{\mu \times n^\prime} $, there exists $ M \in \mathbb{F}_q^{\mu \times n} $ such that $ I \left( \mathbf{x} ; C^\prime N^T \right) = I \left( \mathbf{x} ; C^\prime (M B)^T \right) $.
\end{proposition}
\begin{proof}
Denote by $ \phi $ the security equivalence. Take a matrix $ M \in \mathbb{F}_q^{\mu \times n} $, define $ \mathcal{L} = {\rm Row}(M) \subseteq \mathbb{F}_q^n $ and $ \mathcal{L}^\prime = {\rm Row}(M B) \subseteq \mathbb{F}_q^{n^\prime} $. Then $ \phi(\mathcal{V}_\mathcal{L}) = \mathcal{V}_{\mathcal{L}^\prime} $ and
$$ \dim (\phi(\mathcal{C}_1^\perp) \cap \mathcal{V}_{\mathcal{L}^\prime}) = \dim (\phi(\mathcal{C}_1^\perp \cap \mathcal{V}_\mathcal{L})) = \dim(\mathcal{C}_1^\perp \cap \mathcal{V}_\mathcal{L}), $$ 
and similarly for $ \mathcal{C}_2 $. Thus Equation (\ref{eq security equivalences preserve info}) follows from Proposition \ref{information leakage calculation}.

Observe that we may assume $ n \leq n^\prime $ without loss of generality, since the inverse of a security equivalence is a security equivalence. Thus the injectivity of $ M \mapsto MB $ follows from the fact that $ B $ has full rank. 

Finally, if $ N \in \mathbb{F}_q^{\mu \times n^\prime} $, $ \mathcal{L} = {\rm Row}(N) $ and $ \mathcal{K} = {\rm Row}(B) $, then $ \mathcal{C}_1^\perp \subseteq \mathcal{V}_\mathcal{K} $ and
$$ \mathcal{C}_1^\perp \cap \mathcal{V}_\mathcal{L} = \mathcal{C}_1^\perp \cap ( \mathcal{V}_\mathcal{L} \cap \mathcal{V}_\mathcal{K}), $$
and similarly for $ \mathcal{C}_2^\perp $. Since $ \mathcal{V}_\mathcal{L} \cap \mathcal{V}_\mathcal{K} = \mathcal{V}_{\mathcal{L} \cap \mathcal{K}} $ and $ \mathcal{L} \cap \mathcal{K} = {\rm Row}(MB) $ for a matrix $ M \in \mathbb{F}_q^{\mu \times n} $, the last statement follows again from Proposition \ref{information leakage calculation}.
\end{proof}
 }

The topic of {\color{black}vector space isomorphisms} $ \phi : \mathbb{F}^{m \times n} \longrightarrow \mathbb{F}^{m \times n} $ preserving some specified property has been intensively studied in the literature {\color{black}(see also Table \ref{table isomorphism charact})}, where the term \textit{Frobenius map} is generally used for maps of the form of those in Proposition \ref{morrison proposition}.

When $ m=n $, it is proven in \cite[Theorem 3]{dieudonne} that Frobenius maps are characterized by being those preserving invertible matrices and in \cite{marcus} they are characterized by being those preserving ranks {\color{black}(this is extended to $ m \neq n $ in \cite[Proposition 3]{morrison})}, those preserving determinants and those preserving eigenvalues. 

{\color{black}On the other hand, \cite[Theorem 1]{berger} shows that $ \mathbb{F}_{q^m} $-linear vector space isomorphisms $ \phi : \mathbb{F}_{q^m}^n \longrightarrow \mathbb{F}_{q^m}^n $ preserving ranks are given by $ \phi(\mathbf{c}) = \beta \mathbf{c} A $, for $ \beta \in \mathbb{F}_{q^m} \setminus \{ 0 \} $ and an invertible $ A \in \mathbb{F}_q^{n \times n} $. This is extended in \cite[Theorem 5]{similarities} to $ \mathbb{F}_{q^m} $-linear vector space isomorphisms whose domain and codomain are $ \mathbb{F}_{q^m} $-linear Galois closed spaces in $ \mathbb{F}_{q^m}^n $, which correspond to rank support spaces in $ \mathbb{F}_q^{m \times n} $ (see Lemma \ref{lemma matrix modules are galois} below). 

Therefore, we extend these works in three directions simultaneously: First, we consider the stronger properties (P 1), (P 2) and (P 3) than those considered in \cite{berger, dieudonne, marcus, morrison}, which are essentially (P 4). Second, we extend the domains and codomains from $ \mathbb{F}^{m \times n} $ to general rank support spaces whose matrices do not necessarily have the same sizes. Finally, in the case $ \mathbb{F} = \mathbb{F}_q $, we consider general $ \mathbb{F}_q $-linear maps, instead of the particular case of $ \mathbb{F}_{q^m} $-linear maps as in \cite{berger, similarities}.
}

\subsection{Minimum parameters of linear codes} \label{subsec minimum parameters}

{\color{black}As main application of the previous subsection, we study in this subsection} the minimum parameters $ m $ and $ n $ for which there exists a linear code that is security equivalent to a given one. Recall from Subsection \ref{subsec linear network model} that $ m $ corresponds to the packet length used in the network, and $ n $ corresponds to the number of outgoing links from the source. 

{\color{black}Both cases of one linear code, that is $ \mathcal{C}_2 = \{ 0 \} $ and $ \mathcal{C}_1 = \mathbb{F}^{m \times n} $, are covered since they are dual of each other (see also Remark \ref{remark noiseless coset} and Appendix \ref{app duality theory}). Since security equivalences are rank isometries by Theorem \ref{theorem characterization}, in the first case we find minimum parameters for error and erasure correction, and in the second case we find minimum parameters for universal security on noiseless linearly coded networks. }

\begin{proposition} \label{minimum length proposition}
Fix a linear code $ \mathcal{C} \subseteq \mathbb{F}^{m \times n} $ of dimension $ k $. There exists a linear code $ \mathcal{C}^\prime \subseteq \mathbb{F}^{m \times n^\prime} $ that is security equivalent to $ \mathcal{C} $ if, and only if, $ n^\prime \geq d_{M,k}(\mathcal{C}) $.
\end{proposition}
\begin{proof}
First, if $ \mathcal{C}^\prime \subseteq \mathbb{F}^{m \times n^\prime} $ is security equivalent to $ \mathcal{C} $, then $ \dim(\mathcal{C}^\prime) = k $ and $ d_{M,k}(\mathcal{C}) = d_{M,k}(\mathcal{C}^\prime) \leq n^\prime $.

On the other hand, assume that $ n^\prime \geq d_{M,k}(\mathcal{C}) $. Take a subspace $ \mathcal{L} \subseteq \mathbb{F}^n $ with $ d = \dim(\mathcal{L}) = d_{M,k}(\mathcal{C}) $ and $ \dim(\mathcal{C} \cap \mathcal{V}_\mathcal{L}) \geq k $, which implies that $ \mathcal{C} \subseteq \mathcal{V}_\mathcal{L} $. Take a generator matrix $ A \in \mathbb{F}^{d \times n} $ of $ \mathcal{L} $. There exists a full-rank matrix $ A^\prime \in \mathbb{F}^{n \times d} $ such that $ A A^\prime = I \in \mathbb{F}^{d \times d} $. 

The linear map $ \phi : \mathcal{V}_\mathcal{L} \longrightarrow \mathbb{F}^{m \times d} $, given by $ \phi(V) = V A^\prime $, for $ V \in \mathcal{V}_\mathcal{L} $, is a vector space isomorphism. By dimensions, we just need to see that it is onto. Take $ W \in \mathbb{F}^{m \times d} $. It holds that $ W = WI = WAA^\prime = \phi(WA) $, and $ WA \in \mathcal{V}_\mathcal{L} $ by definition. 

On the other hand, $ \phi $ is a security equivalence by Theorem \ref{theorem characterization}. Therefore $ \phi(\mathcal{C}) \subseteq \mathbb{F}^{m \times d} $ is security equivalent to $ \mathcal{C} $. Finally, we see that appending $ n^\prime - d $ zero columns to the matrices in $ \phi(\mathcal{C}) $ gives a security equivalent code to $ \mathcal{C} $ in $ \mathbb{F}^{m \times n^\prime} $.
\end{proof} 

By transposing matrices, we obtain the following consequence{\color{black}, where we consider linear codes that are rank isometric to a given one. By \cite[Theorem 9]{similarities}, such equivalent codes perform equally when used for error and erasure correction, and by Theorem \ref{maximum secrecy} and Corollary \ref{minimum rank distance}, they perform equally regarding the maximum number of links that an adversary may wire-tap without obtaining any information on noiseless networks. } 

\begin{corollary}
For a linear code $ \mathcal{C} \subseteq \mathbb{F}^{m \times n} $, define the transposed linear code 
$$ \mathcal{C}^T = \{ C^T \mid C \in \mathcal{C} \} \subseteq \mathbb{F}^{n \times m}. $$
If $ m^\prime \geq d_{M,k}(\mathcal{C}^T) $, where $ k = \dim(\mathcal{C}) $, then there exists a linear code $ \mathcal{C}^\prime \subseteq \mathbb{F}^{m^\prime \times n} $ that is rank isometric to $ \mathcal{C} $.
\end{corollary}
\begin{proof}
It follows from Theorem \ref{theorem characterization} and Proposition \ref{minimum length proposition}.
\end{proof}

{\color{black}As a related result, \cite[Proposition 3]{similarities} computes the minimum parameter $ n $ for which there exists an $ \mathbb{F}_{q^m} $-linear code $ \mathcal{C} \subseteq \mathbb{F}_{q^m}^n $ that is rank isometric to a given one. In contrast, we consider both parameters $ m $ and $ n $, we consider security equivalences for the parameter $ n $, and not only rank isometries, and as the biggest difference with \cite{similarities}, we consider general linear codes, and not only $ \mathbb{F}_{q^m} $-linear codes in $ \mathbb{F}_{q^m}^n $. }

\subsection{Degenerate codes} \label{subsec degenerate}

In this subsection, we study degenerate codes, which by the study in the previous subsection, can be applied to networks with less outgoing links or, by transposing matrices, with smaller packet length. {\color{black} Degenerateness of codes in the rank metric has been studied in \cite[Section 6]{slides} and \cite[Subsection IV-B]{similarities}, but only for $ \mathbb{F}_{q^m} $-linear codes in $ \mathbb{F}_{q^m}^n $. We extend those studies to general linear codes in $ \mathbb{F}^{m \times n} $. }

\begin{definition} [\textbf{Degenerate codes}]
We say that a linear code $ \mathcal{C} \subseteq \mathbb{F}^{m \times n} $ is degenerate if it is security equivalent to a linear code $ \mathcal{C}^\prime \subseteq \mathbb{F}^{m \times n^\prime} $ with $ n^\prime < n $.
\end{definition}

The following lemma follows from Proposition \ref{minimum length proposition}:

\begin{lemma}
A linear code $ \mathcal{C} \subseteq \mathbb{F}^{m \times n} $ is degenerate if, and only if, $ d_{M,k}(\mathcal{C}) < n $, where $ k = \dim(\mathcal{C}) $.
\end{lemma}

Now we may give characterizations in terms of the minimum rank distance of the dual code thanks to Proposition \ref{prop duality theorem} {\color{black}in Appendix \ref{app duality theory}}.

\begin{proposition}
Given a linear code $ \mathcal{C} \subseteq \mathbb{F}^{m \times n} $, the following hold:
\begin{enumerate}
\item
Assuming $ \dim(\mathcal{C}^\perp) \geq m $, $ \mathcal{C} $ is degenerate if, and only if, $ d_{M,m}(\mathcal{C}^\perp) = 1 $.
\item
If $ d_R(\mathcal{C}^\perp) > 1 $, then $ \mathcal{C} $ is not degenerate.
\end{enumerate}
\end{proposition}
\begin{proof}
From Proposition \ref{prop duality theorem}, we know that 
$$ \overline{W}_k(\mathcal{C}) \cup W_0(\mathcal{C}^\perp) = \{ 1,2, \ldots, n \}, $$
where the sets on the left-hand side are disjoint, and where $ k = \dim(\mathcal{C}) $. Now, the smallest number in $ \overline{W}_k(\mathcal{C}) $ is $ n+1 - d_{M,k}(\mathcal{C}) $, and the smallest number in $ W_0(\mathcal{C}^\perp) $ is $ d_{M,m}(\mathcal{C}^\perp) $. Item 1 follows from this and the previous lemma. Item 2 follows from item 1 and Proposition \ref{monotonicity of RGMW} {\color{black}in Subsection \ref{subsec monotonicity}}.
\end{proof}

{\color{black}
\section{Monotonicity and Singleton-type bounds} \label{sec singleton}

In this section, we give upper and lower Singleton-type bounds on RGMWs. We start with the monotonicity of RDRPs and RGMWs (Subsection \ref{subsec monotonicity}), which have their own interest, but which are a crucial tool to prove the main bounds (Theorems \ref{th upper singleton} and \ref{lower singleton bound} in Subsection \ref{subsec singleton bounds}). Finally we study linear codes $ \mathcal{C} \subseteq \mathbb{F}^{m \times n} $, meaning $ \mathcal{C}_1 = \mathcal{C} $ and $ \mathcal{C}_2 = \{ 0 \} $, that attain these bounds and whose dimensions are divisible by $ m $ (Subsection \ref{subsec singleton attaining}).

\subsection{Monotonicity of RGMWs and RDRPs} \label{subsec monotonicity}

The monotonicity bounds presented in this subsection are crucial tools for Theorems \ref{th upper singleton} and \ref{lower singleton bound}, but they also have an interpretation in terms of the worst-case information leakage, due to Theorem \ref{maximum secrecy}: An adversary wire-tapping more links in the network will obtain more information in the worst case, and to obtain more information than the worst case for a given number of links, the adversary needs to wire-tap more links. We also bound the corresponding differences.
}

\begin{proposition} [\textbf{Monotonicity of RDRPs}] \label{monotonicity RDRP}
Given nested linear codes $ \mathcal{C}_2 \subsetneqq \mathcal{C}_1 \subseteq \mathbb{F}^{m \times n} $, and $ 0 \leq \mu \leq n-1 $, it holds that $ K_{M,0}(\mathcal{C}_1, \mathcal{C}_2) = 0 $, $ K_{M,n}(\mathcal{C}_1, \mathcal{C}_2) = \dim(\mathcal{C}_1/\mathcal{C}_2) $ and
$$ 0 \leq K_{M,\mu + 1}(\mathcal{C}_1, \mathcal{C}_2) - K_{M,\mu}(\mathcal{C}_1, \mathcal{C}_2) \leq m. $$
\end{proposition}
\begin{proof}
The only property that is not trivial from the definitions is $ K_{M,\mu + 1}(\mathcal{C}_1, \mathcal{C}_2) - K_{M,\mu}(\mathcal{C}_1, \mathcal{C}_2) \leq m $. Consider $ \mathcal{L} \subseteq \mathbb{F}^n $ with $ \dim(\mathcal{L}) \leq \mu + 1 $ and $ \dim(\mathcal{C}_1 \cap \mathcal{V}_\mathcal{L}) - \dim(\mathcal{C}_2 \cap \mathcal{V}_\mathcal{L}) = K_{M,\mu + 1}(\mathcal{C}_1, \mathcal{C}_2) $. 

Take $ \mathcal{L}^\prime \subsetneqq \mathcal{L} $ with $ \dim(\mathcal{L}^\prime) = \dim(\mathcal{L}) - 1 $. Using (\ref{dimension matrix modules}), a simple computation shows that
$$ \dim(\mathcal{C}_1 \cap \mathcal{V}_{\mathcal{L}^\prime}) + m \geq \dim(\mathcal{C}_1 \cap \mathcal{V}_\mathcal{L}). $$
Since $ \dim(\mathcal{C}_2 \cap \mathcal{V}_{\mathcal{L}^\prime}) \leq \dim(\mathcal{C}_2 \cap \mathcal{V}_{\mathcal{L}}) $, it holds that
$$ \dim(\mathcal{C}_1 \cap \mathcal{V}_{\mathcal{L}^\prime}) - \dim(\mathcal{C}_2 \cap \mathcal{V}_{\mathcal{L}^\prime}) + m $$
$$ \geq \dim(\mathcal{C}_1 \cap \mathcal{V}_\mathcal{L}) - \dim(\mathcal{C}_2 \cap \mathcal{V}_{\mathcal{L}}), $$
and the result follows.
\end{proof}

\begin{proposition} [\textbf{Monotonicity of RGMWs}] \label{monotonicity of RGMW}
Given nested linear codes $ \mathcal{C}_2 \subsetneqq \mathcal{C}_1 \subseteq \mathbb{F}^{m \times n} $ with $ \ell = \dim(\mathcal{C}_1 / \mathcal{C}_2) $, it holds that
$$ 0 \leq d_{M,r+1}(\mathcal{C}_1, \mathcal{C}_2) - d_{M,r}(\mathcal{C}_1, \mathcal{C}_2) \leq \min \{ m, n \}, $$ 
for $ 1 \leq r \leq \ell - 1 $, and
$$ d_{M,r}(\mathcal{C}_1, \mathcal{C}_2) + 1 \leq d_{M,r+m}(\mathcal{C}_1, \mathcal{C}_2), $$ 
for $ 1 \leq r \leq \ell - m $.
\end{proposition}
\begin{proof}
The first inequality in the first equation is obvious. We now prove the second inequality. By Proposition \ref{proposition as minimum rank weights}, there exists a subspace $ \mathcal{D} \subseteq \mathcal{C}_1 $ with $ \mathcal{D} \cap \mathcal{C}_2 = \{ 0 \} $, $ \dim(\mathcal{D}) = r $ and $ {\rm wt_R}(\mathcal{D}) = d_{M,r}(\mathcal{C}_1, \mathcal{C}_2) $. Now take $ D \in \mathcal{C}_1 $ not contained in $ \mathcal{D} \oplus \mathcal{C}_2 $, and consider $ \mathcal{D}^\prime = \mathcal{D} \oplus \langle \{ D \} \rangle $. We see from the definitions that $ {\rm RSupp}(\mathcal{D}^\prime) \subseteq {\rm RSupp}(\mathcal{D}) + {\rm Row}(D) $, and hence 
$$ {\rm wt_R}(\mathcal{D}^\prime) \leq {\rm wt_R}(\mathcal{D}) + {\rm Rk}(D) \leq d_{M,r}(\mathcal{C}_1, \mathcal{C}_2) + \min\{m,n\}. $$
Therefore it follows that $ d_{M,r+1}(\mathcal{C}_1, \mathcal{C}_2) \leq d_{M,r}(\mathcal{C}_1, \mathcal{C}_2) + \min\{m,n\} $.

The last inequality follows from Proposition \ref{connection between RDIP RGMW} and Proposition \ref{monotonicity RDRP}.
\end{proof}
 
{\color{black}Due to Theorem \ref{th GMW extend DGW}, the first and third inequalities in the previous proposition coincide with items 3 and 4 in \cite[Theorem 30]{ravagnaniweights} when $ \mathcal{C}_2 = \{ 0 \} $ and $ m \neq n $.
}

{\color{black}
\subsection{Upper and lower Singleton-type bounds} \label{subsec singleton bounds}

Due to Theorem \ref{maximum secrecy}, it is desirable to obtain nested linear code pairs with large RGMWs. The following result gives a fundamental upper bound on them, whose achievability for one linear code ($ \mathcal{C}_2 = \{ 0 \} $) is studied in the next subsection.
}

\begin{theorem}[\textbf{Upper Singleton-type bound}] \label{th upper singleton}
Given nested linear codes $ \mathcal{C}_2 \subsetneqq \mathcal{C}_1 \subseteq \mathbb{F}^{m \times n} $ and $ 1 \leq r \leq \ell = \dim(\mathcal{C}_1/\mathcal{C}_2) $, it holds that
\begin{equation}
d_{M,r}(\mathcal{C}_1,\mathcal{C}_2) \leq n - \left\lceil \frac{\ell - r + 1}{m} \right\rceil + 1.
\label{upper singleton equation}
\end{equation}
In particular, it follows that
$$ \dim(\mathcal{C}_1 / \mathcal{C}_2) \leq \max \{ m,n \}(\min \{ m,n \} - d_R(\mathcal{C}_1,\mathcal{C}_2) + 1), $$
which extends (\ref{singleton bound delsarte}) to nested linear code pairs.
\end{theorem}
\begin{proof}
First of all, we have that $ d_{M,\ell}(\mathcal{C}_1,\mathcal{C}_2) \leq n $ by definition. Therefore the case $ r = \ell $ follows.

For the general case, we will prove that $ m d_{M,r}(\mathcal{C}_1,\mathcal{C}_2) \leq mn - \ell + r + m-1 $. Assume that $ 1 \leq r \leq \ell - hm $, where the integer $ h \geq 0 $ is the maximum possible. That is, $ r + (h+1)m > \ell $. Using Proposition \ref{monotonicity of RGMW}, we obtain
$$ m d_{M,r}(\mathcal{C}_1,\mathcal{C}_2) \leq m d_{M,r + hm}(\mathcal{C}_1,\mathcal{C}_2) - hm $$
$$ \leq m d_{M, \ell}(\mathcal{C}_1,\mathcal{C}_2) - hm \leq mn - \ell + r + m - 1, $$
where the last inequality follows from $ m d_{M,\ell}(\mathcal{C}_1,\mathcal{C}_2) \leq mn $ and $ r + (h+1)m - 1 \geq \ell $.

Finally, the last bound is obtained by setting $ r=1 $ and using Corollary \ref{minimum rank distance} for the given nested linear code pair and the pair obtained by transposing matrices.
\end{proof}

{\color{black}Due to Theorem \ref{th GMW extend DGW}, the previous theorem coincides with item 5 in \cite[Theorem 30]{ravagnaniweights} when $ \mathcal{C}_2 = \{ 0 \} $ and $ m \neq n $.
}

\begin{remark}
In view of \cite[Proposition 1]{rgrw} or \cite[Equation (24)]{luo}, it is natural to wonder whether a sharper bound of the form
$$ d_{M,r}(\mathcal{C}_1,\mathcal{C}_2) \leq n - \left\lceil \frac{\dim(\mathcal{C}_1) - r + 1}{m} \right\rceil + 1 $$
holds. However, this is not the case in general, as the following example shows.
\end{remark}

\begin{example}
Consider $ m = 2 $, the canonical basis $ \mathbf{e}_1, \mathbf{e}_2,$ $ \ldots, $ $ \mathbf{e}_n $ of $ \mathbb{F}^n $, and the linear codes $ \mathcal{C}_1 = \mathbb{F}^{2 \times n} $ and
\begin{displaymath}
\mathcal{C}_2 = \left\langle \left( \begin{array}{c}
\mathbf{e}_1 \\
\mathbf{0}
\end{array} \right), \left( \begin{array}{c}
\mathbf{e}_2 \\
\mathbf{0}
\end{array} \right), \ldots, \left( \begin{array}{c}
\mathbf{e}_n \\
\mathbf{0}
\end{array} \right) \right\rangle.
\end{displaymath}
Observe that $ \ell = \dim(\mathcal{C}_1 / \mathcal{C}_2) = n $. A bound as in the previous remark would imply that $ d_{M,n}(\mathcal{C}_1, \mathcal{C}_2) \leq \lceil n/2 \rceil $. However, a direct inspection shows that $ d_{M,n}(\mathcal{C}_1, \mathcal{C}_2) = n $, since all vectors $ \mathbf{e}_1, \mathbf{e}_2, \ldots, \mathbf{e}_n $ must lie in the row space of any $ \mathcal{D} $ with $ \mathcal{C}_1 = \mathcal{C}_2 \oplus \mathcal{D} $.
\end{example}

On the other hand, we have the following lower bound:

\begin{theorem}[\textbf{Lower Singleton-type bound}] \label{lower singleton bound}
Given nested linear codes $ \mathcal{C}_2 \subsetneqq \mathcal{C}_1 \subseteq \mathbb{F}^{m \times n} $ and $ 1 \leq r \leq \dim(\mathcal{C}_1/\mathcal{C}_2) $, it holds that $ m d_{M,r}(\mathcal{C}_1,\mathcal{C}_2) \geq r $, which implies that
\begin{equation}
d_{M,r}(\mathcal{C}_1,\mathcal{C}_2) \geq \left\lceil \frac{r}{m} \right\rceil.
\label{lower singleton equation}
\end{equation}
\end{theorem}
\begin{proof}
Take a subspace $ \mathcal{D} \subseteq \mathbb{F}^{m \times n} $ and define $ \mathcal{L} = {\rm RSupp}(\mathcal{D}) $. We have that $ \mathcal{D} \subseteq \mathcal{V}_\mathcal{L} $. Using (\ref{dimension matrix modules}), we see that
$$ m {\rm wt_R}(\mathcal{D}) = m \dim(\mathcal{L}) = \dim(\mathcal{V}_\mathcal{L}) \geq \dim(\mathcal{D}). $$
The result follows from this and Proposition \ref{proposition as minimum rank weights}.
\end{proof}

{\color{black}Due to Theorem \ref{th GMW extend DGW}, the previous theorem coincides with item 6 in \cite[Theorem 30]{ravagnaniweights} when $ \mathcal{C}_2 = \{ 0 \} $ and $ m \neq n $.
}

{\color{black}
\subsection{Linear codes attaining the bounds and whose dimensions are divisible by the packet length} \label{subsec singleton attaining}
}

In this subsection, we {\color{black}study the achievability of the bounds (\ref{upper singleton equation}) and (\ref{lower singleton equation}) for one} linear code whose dimension is divisible by the packet length $ m $. {\color{black}As we will show in Subsection \ref{subsec GMW improve DGW}, DGWs \cite{ravagnaniweights} of one linear code coincide with its GMWs when $ m \neq n $. Thus the two propositions below coincide with Corollaries 31 and 32 in \cite{ravagnaniweights} when $ m \neq n $. }

Recall from (\ref{singleton bound delsarte}) that, if a linear code is MRD and $ n \leq m $, then its dimension is divisible by $ m $. In the next proposition, we show that GMWs of MRD linear codes for $ n \leq m $ are all given by $ m $, $ n $ and $ \dim(\mathcal{C}) $, and all attain the upper Singleton-type bound (\ref{upper singleton equation}):

\begin{proposition}
Let $ \mathcal{C} \subseteq \mathbb{F}^{m \times n} $ be a linear code with $ \dim(\mathcal{C}) = mk $. The following are equivalent if $ n \leq m $:
\begin{enumerate}
\item
$ \mathcal{C} $ is maximum rank distance (MRD). 
\item
$ d_{R}(\mathcal{C}) = n - k + 1 $.
\item
$ d_{M,r}(\mathcal{C}) = n - k + \left\lfloor \frac{r - 1}{m} \right\rfloor + 1 $, for all $ 1 \leq r \leq mk $.
\end{enumerate}
\end{proposition}
\begin{proof}
Item 1 and item 2 are equivalent by definition, and item 3 implies item 2 by choosing $ r=1 $. 

Now assume item 2 and let $ 1 \leq r \leq mk $. Let $ r = hm + s $, with $ h \geq 0 $ and $ 0 \leq s < m $. We need to distinguish the cases $ s > 0 $ and $ s = 0 $. We prove only the first case, being the second analogous. By Proposition \ref{monotonicity of RGMW}, we have that
$$ d_{M,r}(\mathcal{C}) \geq h + d_{M,s}(\mathcal{C}) \geq h + d_R(\mathcal{C}) = n - k + h + 1. $$
On the other hand, $ \lceil (mk - r + 1)/m \rceil = k - h $, and therefore the bound (\ref{upper singleton equation}) implies that
$$ d_{M,r}(\mathcal{C}) \leq n - k + h + 1, $$
and hence $ d_{M,r}(\mathcal{C}) = n - k + \lfloor (r-1)/m \rfloor + 1 $ since $ \lfloor (r-1)/m \rfloor = h $, and item 3 follows.
\end{proof}

Regarding the lower Singleton-type bound, we show in the next proposition that rank support spaces are also characterized by having the minimum possible GMWs in view of (\ref{lower singleton equation}):

\begin{proposition}
Let $ \mathcal{C} \subseteq \mathbb{F}^{m \times n} $ be a linear code with $ \dim(\mathcal{C}) = mk $. The following are equivalent:
\begin{enumerate}
\item
$ \mathcal{C} $ is a rank support space. That is, there exists a subspace $ \mathcal{L} \subseteq \mathbb{F}^n $ such that $ \mathcal{C} = \mathcal{V}_\mathcal{L} $.
\item
$ d_{M,km}(\mathcal{C}) = k $.
\item
$ d_{M,r}(\mathcal{C}) = \lceil r/m \rceil $, for all $ 1 \leq r \leq mk $.
\end{enumerate}
\end{proposition}
\begin{proof}
Assume that $ \mathcal{C} = \mathcal{V}_\mathcal{L} $, as in item 1. By taking a sequence of subspaces 
$$ \{ \mathbf{0} \} \subsetneqq \mathcal{L}_1 \subsetneqq \mathcal{L}_2 \subsetneqq \ldots \subsetneqq \mathcal{L}_k = \mathcal{L}, $$
we see that $ d_{M,rm -p}(\mathcal{C}) \leq \dim(\mathcal{L}_r) = r $, for $ 1 \leq r \leq k $ and $ 0 \leq p \leq m-1 $, since $ \dim(\mathcal{C} \cap \mathcal{V}_{\mathcal{L}_r}) = \dim(\mathcal{V}_{\mathcal{L}_r}) = mr \geq mr-p $. Hence item 3 follows.

Item 3 implies item 2 {\color{black}by taking $ r = km $}.

Finally, assume item 2. Take a subspace $ \mathcal{L} \subseteq \mathbb{F}^n $ such that $ \dim(\mathcal{L}) = d_{M,km}(\mathcal{C}) = k $ and $ \dim(\mathcal{C} \cap \mathcal{V}_\mathcal{L}) \geq mk $. By definition and by (\ref{dimension matrix modules}), it holds that $ \dim(\mathcal{C} \cap \mathcal{V}_\mathcal{L}) \geq mk = \dim(\mathcal{V}_\mathcal{L}) $, which implies that $ \mathcal{C} \cap \mathcal{V}_\mathcal{L} = \mathcal{V}_\mathcal{L} $, or in other words, $ \mathcal{V}_\mathcal{L} \subseteq \mathcal{C} $. Since $ \dim(\mathcal{C}) = mk = \dim(\mathcal{V}_\mathcal{L}) $, we see that $ \mathcal{V}_\mathcal{L} = \mathcal{C} $ and item 1 follows.
\end{proof}

\section{Relation with other existing notions of generalized weights} \label{sec relation with others}

In this section, we study the relation between RGMWs and RDRPs and other notions of generalized weights {\color{black}(see Table \ref{table notions gen weights}). We first show that RGMWs and RDRPs extend RGRWs and RDIPs \cite{rgrw, oggier} (Theorem \ref{th RGMW extend RGRW} in Subsection \ref{subsec RGMW extend RGRW}), respectively, then we show that they extend RGHWs and RDLPs \cite{forney, luo, wei} (Theorem \ref{th RGMW extend RGHW} in Subsection \ref{subsec RGMW extend RGHW}), respectively, and we conclude by showing that GMWs coincide with DGWs \cite{ravagnaniweights} for one linear code, meaning $ \mathcal{C}_1 = \mathcal{C} $ arbitrary and $ \mathcal{C}_2 = \{ 0 \} $, when $ m \neq n $, and are strictly larger when $ m = n $ (Theorem \ref{th GMW extend DGW} in Subsection \ref{subsec GMW improve DGW}).
}

\subsection{RGMWs extend relative generalized rank weights} \label{subsec RGMW extend RGRW}

In this subsection, we prove that RGMWs and RDRPs extend RGRWs and RDIPs \cite{rgrw, oggier}, respectively.
%
%Throughout the subsection, we will consider the extension field $ \mathbb{F}_{q^m} $ of the finite field $ \mathbb{F}_q $, and vector spaces in $ \mathbb{F}_{q^m}^n $ will be considered to be linear over $ \mathbb{F}_{q^m} $. We need the notion of Galois closed spaces \cite{stichtenoth}:

\begin{definition} [\textbf{Galois closed spaces \cite{stichtenoth}}]
We say that an $ \mathbb{F}_{q^m} $-linear vector space $ \mathcal{V} \subseteq \mathbb{F}_{q^m}^n $ is Galois closed if 
\begin{equation*}
\mathcal{V}^q = \{ (v_1^q, v_2^q, \ldots, v_n^q) \mid (v_1, v_2, \ldots, v_n) \in \mathcal{V} \} \subseteq \mathcal{V}.
%\label{Galois closed equation}
\end{equation*}
We denote by $ \Upsilon(\mathbb{F}_{q^m}^n) $ the family of $ \mathbb{F}_{q^m} $-linear Galois closed vector spaces in $ \mathbb{F}_{q^m}^n $.
\end{definition}

RGRWs and RDIPs are then defined in \cite{rgrw} as follows:

\begin{definition} [\textbf{Relative Generalized Rank Weigths \cite[Definition 2]{rgrw}}]
Given nested $ \mathbb{F}_{q^m}$-linear codes $ \mathcal{C}_2 \subsetneqq \mathcal{C}_1 \subseteq \mathbb{F}_{q^m}^n $, and $ 1 \leq r \leq \ell = \dim(\mathcal{C}_1 / \mathcal{C}_2) $ (over $ \mathbb{F}_{q^m} $), we define their $ r $-th relative generalized rank weight (RGRW) as
\begin{equation*}
\begin{split}
d_{R,r}(\mathcal{C}_1, \mathcal{C}_2) = \min \{ & \dim(\mathcal{V}) \mid \mathcal{V} \in \Upsilon(\mathbb{F}_{q^m}^n), \\
 & \dim(\mathcal{C}_1 \cap \mathcal{V}) - \dim(\mathcal{C}_2 \cap \mathcal{V}) \geq r \},
\end{split}
%\label{RGRW}
\end{equation*}
where dimensions are taken over $ \mathbb{F}_{q^m} $.
\end{definition}

\begin{definition} [\textbf{Relative Dimension/Intersection Profile \cite[Definition 1]{rgrw}}]
Given nested $ \mathbb{F}_{q^m} $-linear codes $ \mathcal{C}_2 \subsetneqq \mathcal{C}_1 \subseteq \mathbb{F}_{q^m}^n $, and $ 0 \leq \mu \leq n $, we define their $ \mu $-th relative dimension/intersection profile (RDIP) as
\begin{equation*}
\begin{split}
K_{R,\mu}(\mathcal{C}_1, \mathcal{C}_2) = \max \{ & \dim(\mathcal{C}_1 \cap \mathcal{V}) - \dim(\mathcal{C}_2 \cap \mathcal{V}) \mid \\
 & \mathcal{V} \in \Upsilon(\mathbb{F}_{q^m}^n), \dim(\mathcal{V}) \leq \mu \},
\end{split}
%\label{RDIP linear}
\end{equation*}
where dimensions are taken over $ \mathbb{F}_{q^m} $.
\end{definition}

{\color{black}The following is the main result of the subsection, which shows that Theorem \ref{maximum secrecy} extends the study on worst-case information leakage on $ \mathbb{F}_q $-linearly coded networks in \cite{rgrw} (see its Theorem 2 and Corollary 5) from $ \mathbb{F}_{q^m} $-linear codes in $ \mathbb{F}_{q^m}^n $ to general $ \mathbb{F}_q $-linear codes in $ \mathbb{F}_q^{m \times n} $, when considering uniform probability distributions. }

\begin{theorem} \label{th RGMW extend RGRW}
Let $ \alpha_1, \alpha_2, \ldots, \alpha_m $ be a basis of $ \mathbb{F}_{q^m} $ as a vector space over $ \mathbb{F}_q $. Given nested $ \mathbb{F}_{q^m} $-linear codes $ \mathcal{C}_2 \subsetneqq \mathcal{C}_1 \subseteq \mathbb{F}_{q^m}^n $, and integers $ 1 \leq r \leq \ell = \dim(\mathcal{C}_1 / \mathcal{C}_2) $ (over $ \mathbb{F}_{q^m} $), $ 0 \leq p \leq m-1 $ and $ 0 \leq \mu \leq n $, we have that
\begin{equation*}
d_{R,r}(\mathcal{C}_1, \mathcal{C}_2) = d_{M,rm - p}(M_{\boldsymbol\alpha}(\mathcal{C}_1), M_{\boldsymbol\alpha}(\mathcal{C}_2)),
\end{equation*}
\begin{equation*}
m K_{R,\mu}(\mathcal{C}_1, \mathcal{C}_2) = K_{M,\mu}(M_{\boldsymbol\alpha}(\mathcal{C}_1), M_{\boldsymbol\alpha}(\mathcal{C}_2)),
\end{equation*}
{\color{black}where $ M_{\boldsymbol\alpha} : \mathbb{F}_{q^m}^n \longrightarrow \mathbb{F}_q^{m \times n} $ is as in (\ref{equation matrix representation}). }
%Moreover, it holds that $ \mathcal{I}_{M,rm - p}(M_{\boldsymbol\alpha}(\mathcal{C}_1), M_{\boldsymbol\alpha}(\mathcal{C}_2)) = \emptyset $ if $ p \neq 0 $, and
%$$ \mathcal{I}_{M,rm}(M_{\boldsymbol\alpha}(\mathcal{C}_1), M_{\boldsymbol\alpha}(\mathcal{C}_2)) = \{ \mathcal{L} \subseteq \mathbb{F}_q^n \mid $$
%$$ \dim(\mathcal{C}_2^\perp \cap M_{\boldsymbol\alpha}^{-1}(\mathcal{V}_\mathcal{L})) - \dim(\mathcal{C}_1^\perp \cap M_{\boldsymbol\alpha}^{-1}(\mathcal{V}_\mathcal{L})) = r \}, $$
%where dimensions in the last equality are taken over $ \mathbb{F}_{q^m} $.
\end{theorem}

{\color{black}The theorem follows from the next two lemmas, where we take the first one from \cite{stichtenoth}: }

\begin{lemma} [\textbf{\cite[Lemma 1]{stichtenoth}}] \label{lemma stichtenoth}
An $ \mathbb{F}_{q^m} $-linear vector space $ \mathcal{V} \subseteq \mathbb{F}_{q^m}^n $ is Galois closed if, and only if, it has a basis of vectors in $ \mathbb{F}_q^n $ as a vector space over $ \mathbb{F}_{q^m} $. 
\end{lemma}

\begin{lemma} \label{lemma matrix modules are galois}
Let $ \alpha_1, \alpha_2, \ldots, \alpha_m $ be a basis of $ \mathbb{F}_{q^m} $ as a vector space over $ \mathbb{F}_q $, and let $ \mathcal{V} \subseteq \mathbb{F}_{q^m}^n $ be an arbitrary set. The following are equivalent:
\begin{enumerate}
\item
$ \mathcal{V} \subseteq \mathbb{F}_{q^m}^n $ is an $ \mathbb{F}_{q^m} $-linear Galois closed vector space. That is, $ \mathcal{V} \in \Upsilon(\mathbb{F}_{q^m}^n) $.
\item
$ M_{\boldsymbol\alpha}(\mathcal{V}) \subseteq \mathbb{F}_q^{m \times n} $ is a rank support space. That is, $  M_{\boldsymbol\alpha}(\mathcal{V}) \in RS(\mathbb{F}_q^{m \times n}) $.
\end{enumerate}
Moreover, if $ M_{\boldsymbol\alpha}(\mathcal{V}) = \mathcal{V}_\mathcal{L} $ for a subspace $ \mathcal{L} \subseteq \mathbb{F}_q^n $, then
$$ \dim(\mathcal{V}) = \dim(\mathcal{L}), $$
where $ \dim(\mathcal{V}) $ is taken over $ \mathbb{F}_{q^m} $ and $ \dim(\mathcal{L}) $ over $ \mathbb{F}_q $.
\end{lemma}
\begin{proof}
We first observe the following. For an arbitrary set $ \mathcal{V} \subseteq \mathbb{F}_{q^m}^n $, the previous lemma states that $ \mathcal{V} $ is an $ \mathbb{F}_{q^m} $-linear Galois closed vector space if, and only if, $ \mathcal{V} $ is $ \mathbb{F}_q $-linear and it has a basis over $ \mathbb{F}_q $ of the form $ \mathbf{v}_{i,j} = \alpha_i \mathbf{b}_j $, for $ i = 1,2, \ldots, m $ and $ j = 1,2, \ldots, k $, where $ \mathbf{b}_1, \mathbf{b}_2, \ldots, \mathbf{b}_k \in \mathbb{F}_q^n $. By considering $ B_{i,j} = M_{\boldsymbol\alpha}(\mathbf{v}_{i,j}) \in \mathbb{F}_q^{m \times n} $, we see that this condition is equivalent to item 2 in Proposition \ref{prop characterization}, and we are done.
\end{proof}

{\color{black}
\begin{remark}
The results in this subsection can be extended to Galois extensions of fields $ \mathbb{F} \subseteq \widetilde{\mathbb{F}} $ of finite degree $ m $. For that purpose, we only need to define Galois closed spaces as those $ \widetilde{\mathbb{F}} $-linear subspaces $ \mathcal{V} \subseteq \widetilde{\mathbb{F}}^n $ that are closed under the action of every field morphism in the Galois group of the extension $ \mathbb{F} \subseteq \widetilde{\mathbb{F}} $. The rest of definitions and results in this subsection can be directly translated word by word to this case, except for Lemma \ref{lemma stichtenoth}, which would be replaced by \cite[Theorem 1]{galoisinvariance}.

Thus the results in this subsection can be applied to generalizations of rank-metric codes such as those in \cite{augot}.
\end{remark}
}

\subsection{RGMWs extend relative generalized Hamming weights} \label{subsec RGMW extend RGHW}

In this subsection, we show that RGMWs and RDRPs also extend RGHWs and RDLPs \cite{forney, luo, wei}, respectively. We start with the definitions of Hamming supports and Hamming support spaces:

\begin{definition} [\textbf{Hamming supports}]
Given a vector space $ \mathcal{C} \subseteq \mathbb{F}^n $, we define its Hamming support as
\begin{equation*}
\begin{split}
{\rm HSupp}(\mathcal{C}) = \{ & i \in \{ 1,2, \ldots, n \} \mid \\
 & \exists (c_1, c_2, \ldots, c_n) \in \mathcal{C}, c_i \neq 0 \}.
\end{split}
%\label{Hamming support}
\end{equation*}
We also define the Hamming weight of the space $ \mathcal{C} $ as 
\begin{equation*}
{\rm wt_H}(\mathcal{C}) = |{\rm HSupp}(\mathcal{C}) |.
%\label{Hamming weigth}
\end{equation*}
Finally, for a vector $ \mathbf{c} \in \mathbb{F}^n $, we define its Hamming support as $ {\rm HSupp}(\mathbf{c}) = {\rm HSupp}(\langle \{ \mathbf{c} \} \rangle) $, and its Hamming weight as $ {\rm wt_H}(\mathbf{c}) = {\rm wt_H}(\langle \{ \mathbf{c} \} \rangle) $.
\end{definition}

\begin{definition} [\textbf{Hamming support spaces}]
Given a subset $ I \subseteq \{ 1,2, \ldots, n \} $, we define its Hamming support space as the vector space in $ \mathbb{F}^n $ given by
\begin{equation*}
\mathcal{L}_I = \{ (c_1, c_2, \ldots, c_n) \in \mathbb{F}^n \mid c_i = 0, \forall i \notin I \}.
%\label{Hamming support space equation}
\end{equation*}
\end{definition}

We may now define RGHWs and RDLPs:

\begin{definition} [\textbf{Relative Generalized Hamming Weigths \cite[Section III]{luo}}]
Given nested linear codes $ \mathcal{C}_2 \subsetneqq \mathcal{C}_1 \subseteq \mathbb{F}^n $, and $ 1 \leq r \leq \ell = \dim(\mathcal{C}_1 / \mathcal{C}_2) $, we define their $ r $-th relative generalized Hamming weight (RGHW) as
\begin{equation*}
\begin{split}
d_{H,r}(\mathcal{C}_1, \mathcal{C}_2) = \min \{ & | I | \mid I \subseteq \{ 1,2, \ldots, n \}, \\
 & \dim(\mathcal{C}_1 \cap \mathcal{L}_I) - \dim(\mathcal{C}_2 \cap \mathcal{L}_I) \geq r \}.
\end{split}
%\label{RGHW}
\end{equation*}
\end{definition}

As in Proposition \ref{proposition as minimum rank weights}, it holds that
\begin{equation*}
\begin{split}
d_{H,r}(\mathcal{C}_1, \mathcal{C}_2) = \min \{ & {\rm wt_H}(\mathcal{D}) \mid \mathcal{D} \subseteq \mathcal{C}_1, \mathcal{D} \cap \mathcal{C}_2 = \{ 0 \}, \\
 & \dim(\mathcal{D}) = r \}.
\end{split}
%\label{RGHW alternative definition}
\end{equation*}

Given a linear code $ \mathcal{C} \subseteq \mathbb{F}^n $, we see that its $ r $-th GHW \cite[Section II]{wei} is $ d_{H,r}(\mathcal{C}) = d_{H,r}(\mathcal{C}, \{ \mathbf{0} \}) $, for $ 1 \leq r \leq \dim(\mathcal{C}) $.

\begin{definition} [\textbf{Relative Dimension/Length Profile \cite{forney, luo}}]
Given nested linear codes $ \mathcal{C}_2 \subsetneqq \mathcal{C}_1 \subseteq \mathbb{F}^n $, and $ 0 \leq \mu \leq n $, we define their $ \mu $-th relative dimension/length profile (RDLP) as
\begin{equation*}
\begin{split}
K_{H,\mu}(\mathcal{C}_1, \mathcal{C}_2) = \max \{ & \dim(\mathcal{C}_1 \cap \mathcal{L}_I) - \dim(\mathcal{C}_2 \cap \mathcal{L}_I) \mid \\
 & I \subseteq \{ 1,2, \ldots, n \}, | I | \leq \mu \}.
\end{split}
%\label{RDLP}
\end{equation*}
\end{definition}

To prove {\color{black}our results}, we need to see vectors in $ \mathbb{F}^n $ as matrices in $ \mathbb{F}^{n \times n} $. To that end, we introduce the diagonal matrix representation map $ \Delta : \mathbb{F}^n \longrightarrow \mathbb{F}^{n \times n} $ given by
\begin{equation}
\Delta (\mathbf{c}) = {\rm diag}(\mathbf{c}) = (c_i \delta_{i,j})_{1 \leq i \leq n, 1 \leq j \leq n},
\end{equation}
where $ \mathbf{c} = (c_1, c_2, \ldots, c_n) \in \mathbb{F}^n $ and $ \delta_{i,j} $ represents the Kronecker delta. In other words, $ \Delta (\mathbf{c}) $ is the diagonal matrix whose diagonal vector is $ \mathbf{c} $.

The map $ \Delta : \mathbb{F}^n \longrightarrow \mathbb{F}^{n \times n} $ is linear, one to one {\color{black}and, for any vector space $ \mathcal{D} \subseteq \mathbb{F}^n $, it holds that
$$ {\rm wt_R}(\Delta(\mathcal{D})) = {\rm wt_H}(\mathcal{D}). $$
}

{\color{black}We may now give the main result of this subsection: }

\begin{theorem} \label{th RGMW extend RGHW}
Given nested linear codes $ \mathcal{C}_2 \subsetneqq \mathcal{C}_1 \subseteq \mathbb{F}^n $, and integers $ 1 \leq r \leq \ell = \dim(\mathcal{C}_1 / \mathcal{C}_2) $, and $ 0 \leq \mu \leq n $, we have that
\begin{equation*}
d_{H,r}(\mathcal{C}_1, \mathcal{C}_2) = d_{M,r}(\Delta(\mathcal{C}_1), \Delta(\mathcal{C}_2)),
\end{equation*}
\begin{equation*}
K_{H,\mu}(\mathcal{C}_1, \mathcal{C}_2) = K_{M,\mu}(\Delta(\mathcal{C}_1), \Delta(\mathcal{C}_2)).
\end{equation*}
%Moreover, it holds that
%$$ \mathcal{I}_{M,r}(\Delta(\mathcal{C}_1), \Delta(\mathcal{C}_2)) = \{ \mathcal{L}_I \subseteq \mathbb{F}^n \mid $$
%$$ I \subseteq \{ 1,2, \ldots, n \}, \dim(\mathcal{C}_2^\perp \cap \mathcal{L}_I) - \dim(\mathcal{C}_1^\perp \cap \mathcal{L}_I) = r \}. $$
\end{theorem}
{\color{black}\begin{proof}
We prove the first equality, being the second analogous. Denote by $ d_r $ the number on the left-hand side and by $ d_r^\prime $ the number on the right-hand side, and prove both inequalities:

$ d_r \leq d_r^\prime $: Take a vector space $ \mathcal{L} \subseteq \mathbb{F}^n $ such that $ \dim(\mathcal{L}) = d_r^\prime $ and $ \dim((\Delta(\mathcal{C}_1) \cap \mathcal{V}_\mathcal{L}) / (\Delta(\mathcal{C}_2) \cap \mathcal{V}_\mathcal{L})) \geq r $. It holds that $ \mathcal{V}_\mathcal{L} \cap \Delta(\mathbb{F}^n) = \Delta(\mathcal{L}_I) $, for some subset $ I \subseteq \{ 1,2, \ldots, n \} $. We have that $ \dim((\mathcal{C}_1 \cap \mathcal{L}_I) / (\mathcal{C}_2 \cap \mathcal{L}_I)) \geq r $ and
$$ d_r \leq | I | = {\rm wt_R}(\Delta(\mathcal{L}_I)) \leq {\rm wt_R}(\mathcal{V}_\mathcal{L}) = \dim(\mathcal{L}) = d_r^\prime. $$

$ d_r \geq d_r^\prime $: Take a subset $ I \subseteq \{ 1,2, \ldots, n \} $ such that $ | I | = d_r $ and $ \dim((\mathcal{C}_1 \cap \mathcal{L}_I) / (\mathcal{C}_2 \cap \mathcal{L}_I)) \geq r $. Now it holds that $ \Delta(\mathcal{L}_I) = \mathcal{V}_{\mathcal{L}_I} \cap \Delta(\mathbb{F}^n) $. Therefore $ \dim((\Delta(\mathcal{C}_1) \cap \mathcal{V}_{\mathcal{L}_I}) / (\Delta(\mathcal{C}_2) \cap \mathcal{V}_{\mathcal{L}_I})) \geq r $ and
$$ d_r^\prime \leq \dim(\mathcal{L}_I) = | I | = d_r. $$
\end{proof}
}

%Observe that the last equality shows that the access structures of ramp secret sharing schemes, constructed using nested linear code pairs, as defined in \cite[Section II]{one-point}, are also extended by the matrix information distributions in Definition \ref{def information distributions}.
%
%On the other hand, just as in \cite[Theorem 7]{similarities}, the previous equalities imply that every bound on relative generalized matrix weights must be ``less tight'' than the corresponding bounds on relative generalized Hamming weights:
%
%\begin{corollary} 
%Fix numbers $ \ell $ and $ 1 \leq r, s \leq \ell $, and functions $ f_{r,s}, g_{r,s} : \mathbb{N} \longrightarrow \mathbb{R} $, which may also depend on $ n, m, \ell $ and $ q = \# \mathbb{F} $ if $ \mathbb{F} $ is finite. Every bound of the form 
%$$ f_{r,s}(d_r (\mathcal{C}_1, \mathcal{C}_2)) \geq g_{r,s}(d_s (\mathcal{C}_1, \mathcal{C}_2)) $$
%that is valid for relative generalized matrix weights, for any pair of linear codes $ \mathcal{C}_2 \subsetneqq \mathcal{C}_1 \subseteq \mathbb{F}^{n \times n} $ with $ \dim(\mathcal{C}_1 / \mathcal{C}_2) = \ell $, is also valid for relative generalized Hamming weights of a pair of linear codes $ \mathcal{D}_2 \subsetneqq \mathcal{D}_1 \subseteq \mathbb{F}^n $ with $ \dim(\mathcal{D}_1 / \mathcal{D}_2) = \ell $. The same holds for generalized weights of just one linear code $ \mathcal{C} \subseteq \mathbb{F}^{n \times n} $ and $ \mathcal{D} \subseteq \mathbb{F}^n $.
%\end{corollary}

\subsection{Relation with Delsarte generalized weights} \label{subsec GMW improve DGW}

A notion of generalized weights, called Delsarte generalized weights (DGWs), for a linear code{\color{black}, which in this section means $ \mathcal{C}_1 = \mathcal{C} $ arbitrary and $ \mathcal{C}_2 = \{ 0 \} $} has already been proposed in \cite{ravagnaniweights} as an algebraic invariant of the code. We will prove that GMWs are strictly larger than DGWs when {\color{black}$ m = n $}, and we will prove that both coincide in the other cases.

These weights are defined in terms of optimal anticodes for the rank metric:

\begin{definition} [\textbf{Maximum rank distance}]
For a linear code $ \mathcal{C} \subseteq \mathbb{F}^{m \times n} $, we define its maximum rank distance as 
$$ {\rm MaxRk}(\mathcal{C}) = \max \{ {\rm Rk (C)} \mid C \in \mathcal{C}, C \neq 0 \}. $$
\end{definition}

The following bound is given in \cite[Proposition 47]{ravagnani}:
\begin{equation} \label{anticode bound}
\dim(\mathcal{C}) \leq m {\rm MaxRk}(\mathcal{C}).
\end{equation}

This leads to the definition of rank-metric optimal anticodes:

\begin{definition}[\textbf{Optimal anticodes \cite[Definition 22]{ravagnaniweights}}]
We say that a linear code $ \mathcal{V} \subseteq \mathbb{F}^{m \times n} $ is a (rank-metric) optimal anticode if equality in (\ref{anticode bound}) holds. 

We will denote by $ A(\mathbb{F}^{m \times n}) $ the family of linear optimal anticodes in $ \mathbb{F}^{m \times n} $.
\end{definition}

In view of this, DGWs are defined in \cite{ravagnaniweights} as follows:

\begin{definition}[\textbf{Delsarte generalized weights \cite[Definition 23]{ravagnaniweights}}] \label{DGW}
For a linear code $ \mathcal{C} \subseteq \mathbb{F}^{m \times n} $ and an integer $ 1 \leq r \leq \dim(\mathcal{C}) $, we define its $ r $-th Delsarte generalized weight (DGW) as
\begin{equation*}
\begin{split}
d_{D,r}(\mathcal{C}) = m^{-1} \min \{ & \dim(\mathcal{V}) \mid \mathcal{V} \in A(\mathbb{F}^{m \times n}), \\
 & \dim(\mathcal{C} \cap \mathcal{V}) \geq r \}.
\end{split}
%\label{def2}
\end{equation*}
\end{definition}

Observe that $ d_{D,r}(\mathcal{C}) $ is an integer since the dimension of optimal anticodes is a multiple of $ m $ by definition.

{\color{black}Before giving the main result, we need the following proposition: }

\begin{proposition} \label{proposition ravagnani}
If a set $ \mathcal{V} \subseteq \mathbb{F}^{m \times n} $ is a rank support space, then it is a (rank-metric) optimal anticode. In other words, $ RS(\mathbb{F}^{m \times n}) \subseteq A(\mathbb{F}^{m \times n}) $. The reversed inclusion also holds if $ m \neq n $.
\end{proposition}
\begin{proof}
We first prove that $ RS(\mathbb{F}^{m \times n}) \subseteq A(\mathbb{F}^{m \times n}) $. Let $ \mathcal{V} \in RS(\mathbb{F}^{m \times n}) $ and let $ B_{i,j} $, $ i = 1,2, \ldots, m $ and $ j = 1,2, \ldots, k $, be a basis of $ \mathcal{V} $ as in Proposition \ref{prop characterization}, item 2. For any $ V = \sum_{i=1}^m \sum_{j=1}^k \lambda_{i,j} B_{i,j} \in \mathcal{V} $, with $ \lambda_{i,j} \in \mathbb{F} $, it holds that 
$$ {\rm Rk}(V) \leq \dim( \langle \mathbf{b}_1, \mathbf{b}_2, \ldots, \mathbf{b}_k \rangle ) = k, $$
where $ \mathbf{b}_1, \mathbf{b}_2, \ldots, \mathbf{b}_k $ are as in Proposition \ref{prop characterization}, item 2. Therefore $ \dim(\mathcal{V}) = mk \geq m {\rm MaxRk}(\mathcal{V}) $ and $ \mathcal{V} $ is an optimal anticode.

We now prove that $  A(\mathbb{F}^{m \times n}) \subseteq RS(\mathbb{F}^{m \times n}) $ when $ m \neq n $. Let $ \mathcal{V} \in A(\mathbb{F}^{m \times n}) $. By \cite[Theorem 26]{ravagnaniweights}, there exist full-rank matrices $ A \in \mathbb{F}_q^{m \times m} $ and $ B \in \mathbb{F}_q^{n \times n} $ such that $ \mathcal{V} = \{ ACB \in \mathbb{F}_q^{m \times n} \mid C \in \mathcal{V}_\mathcal{L} \} $, where $ \mathcal{L} = \mathbb{F}_q^k \times \{ 0 \}^{n-k} $ for some positive integer $ k $. By Proposition \ref{prop characterization}, $ \mathcal{V} $ is a rank support space and we are done.
\end{proof}

{\color{black}In \cite[Theorem 18]{ravagnaniweights} it is proven that $ \mathcal{V} \subseteq \mathbb{F}_{q^m}^n $ is an $ \mathbb{F}_{q^m} $-linear Galois closed vector space if, and only if, it is an $ \mathbb{F}_{q^m} $-linear vector space satisfying equality in (\ref{anticode bound}). Hence due to Lemma \ref{lemma matrix modules are galois}, the previous proposition strengthens \cite[Theorem 18]{ravagnaniweights} when $ m \neq n $ by showing that the $ \mathbb{F}_{q^m} $-linearity of $ \mathcal{V} $ may be weakened to $ \mathbb{F}_q $-linearity. Moreover, our result holds for any field $ \mathbb{F} \neq \mathbb{F}_q $. 

The main result of this subsection is the next theorem, which follows from the previous proposition and the corresponding definitions: }

\begin{theorem} \label{th GMW extend DGW}
For a linear code $ \mathcal{C} \subseteq \mathbb{F}^{m \times n} $ and an integer $ 1 \leq r \leq \dim(\mathcal{C}) $, we have that
\begin{equation*}
\begin{split}
d_{D,r}(\mathcal{C}) & \leq d_{M,r}(\mathcal{C}) \textrm{ if } m=n, \textrm{ and} \\
d_{D,r}(\mathcal{C}) & = d_{M,r}(\mathcal{C}) \textrm{ if } m \neq n.
\end{split}
\end{equation*}
\end{theorem}

{\color{black}Due to Theorem \ref{maximum secrecy}, when considering universal security on linearly coded networks it is desirable to obtain linear codes $ \mathcal{C} \subseteq \mathbb{F}_q^{m \times n} $ with large GMWs. Therefore, linear codes with large DGWs serve this purpose, but linear codes with low DGWs may still have large GMWs when $ m = n $. }

The next example shows that not all linear optimal anticodes are rank support spaces when $ m=n $, that is, $ RS(\mathbb{F}^{n \times n}) \subsetneqq A(\mathbb{F}^{n \times n}) $, for any $ n $ and any field $ \mathbb{F} $. As a consequence, in some cases GMWs are strictly larger than DGWs. To that end, we will use the characterization of rank support spaces as matrix modules from Appendix \ref{app matrix modules}.

\begin{example} \label{example delsarte less than}
Consider $ m = n = 2 $ and the linear code
\begin{displaymath}
\mathcal{C} = \left\langle \left( \begin{array}{cc}
1 & 0 \\
0 & 0
\end{array} \right), \left( \begin{array}{cc}
0 & 1 \\
0 & 0
\end{array} \right) \right\rangle \subseteq \mathbb{F}^{2 \times 2}.
\end{displaymath}
It holds that $ \dim(\mathcal{C}) = 2 $, $ m = 2 $ and $ {\rm MaxRk}(\mathcal{C}) = 1 $. Therefore $ \mathcal{C} $ is an optimal anticode. However, it is not a matrix module, and therefore it is not a rank support space (see Appendix \ref{app matrix modules}), since
\begin{displaymath}
\left( \begin{array}{cc}
0 & 0 \\
1 & 0
\end{array} \right) \left( \begin{array}{cc}
1 & 0 \\
0 & 0
\end{array} \right) = \left( \begin{array}{cc}
0 & 0 \\
1 & 0
\end{array} \right) \notin \mathcal{C}.
\end{displaymath}
In other words, $ RS(\mathbb{F}^{2 \times 2}) \subsetneqq A(\mathbb{F}^{2 \times 2}) $.

On the one hand, we have that $ d_{D,1}(\mathcal{C}) = d_{D,2}(\mathcal{C}) = 1 $, by \cite[Corollary 32]{ravagnaniweights}, or just by inspection.

On the other hand, it is easy to check that $ d_{M,1}(\mathcal{C}) = 1 $, and since $ {\rm RSupp}(\mathcal{C}) = \mathbb{F}^2 $, it holds that $ d_{M,2}(\mathcal{C}) = 2 $. Therefore $ d_{M,2}(\mathcal{C}) > d_{D,2}(\mathcal{C}) $.

Observe that we may trivially extend this example to any value of $ m=n $, and it holds for an arbitrary field $ \mathbb{F} $.
\end{example}

{\color{black}
\section{Conclusion and open problems}

In this work, we have extended the study of universal security provided by $ \mathbb{F}_{q^m} $-linear nested coset coding schemes from \cite{rgrw, silva-universal} to that provided by $ \mathbb{F}_q $-linear schemes, where $ \mathbb{F}_q $ is the field used on the network and $ m $ is the packet length. 

Thanks to this study, we have completed the list of parameters $ \ell $, $ t $, $ m $ and $ n $ for which we can obtain optimal universal secure $ \mathbb{F}_q $-linear codes for noiseless networks from \cite{silva-universal}, and we have added \textit{near optimal} universal security to the rank list-decodable codes from \cite{list-decodable-rank-metric}, providing the first universal secure linear coset coding schemes able to list-decode in polynomial time roughly twice the rank errors that optimal universal secure schemes can unique-decode, with almost the same secret message size $ \ell $ and security parameter $ t $. 

Motivated by our study, we defined a family of security equivalences between linear coset coding schemes and gave mathematical characterizations of such equivalences, which allowed us to obtain, in terms of the last generalized matrix weight, ranges of parameters $ m $ and $ n $ of networks on which a linear code can be applied with the same security performance.

Finally, we give the following list of open problems:

1) Obtain optimal universal secure and error-correcting linear coset coding schemes for noisy networks for all possible parameters $ \ell $, $ t $, $ m $, $ n $, and number of rank errors. 

2) Extend the concept of universal \textit{strong security} from \cite[Definition 6]{rgrw} to general $ \mathbb{F}_q $-linear coset coding schemes, and provide optimal universal strong secure schemes as those in \cite[Section V]{rgrw} for all possible parameters $ \ell $, $ t $, $ m $ and $ n $, for either noiseless or noisy networks.

3) Subsection \ref{subsec near optimality list decoding} implies that $ \ell $ is close to but smaller than $ n - t - e $, where $ e $ is the number of list-decodable rank errors with polynomial list sizes $ L $. We conjecture, but leave as open problem, that a bound similar to $ \ell \leq n - t - e $ holds in general.

4) Study the sharpness of the bounds given in Theorem \ref{lower singleton bound}.  }

\appendices

{\color{black}
\section{Duality theory} \label{app duality theory}

In this appendix, we collect technical results concerning trace duality of linear codes in $ \mathbb{F}^{m \times n} $ used throughout the paper. Some of the results are taken or expanded from the literature, and some are new. Recall first the definition of trace product and dual of a linear code in $ \mathbb{F}^{m \times n} $ (Definition \ref{def trace inner product}).
}

First, since the trace product in $ \mathbb{F}^{m \times n} $ coincides with the usual inner product in $ \mathbb{F}^{mn} $, it holds that
$$ \dim(\mathcal{C}^\perp) = mn - \dim(\mathcal{C}), \quad \mathcal{C} \subseteq \mathcal{D} \Longleftrightarrow \mathcal{D}^\perp \subseteq \mathcal{C}^\perp, $$
$$ \mathcal{C}^{\perp \perp} = \mathcal{C}, \quad (\mathcal{C} + \mathcal{D})^\perp = \mathcal{C}^\perp \cap \mathcal{D}^\perp, \quad (\mathcal{C} \cap \mathcal{D})^\perp = \mathcal{C}^\perp + \mathcal{D}^\perp, $$
for linear codes $ \mathcal{C}, \mathcal{D} \subseteq \mathbb{F}^{m \times n} $. {\color{black}We have the following:}

\begin{lemma}[\textbf{\cite[Lemma 27]{ravagnani}}] \label{lemma dual rank support}
If $ \mathcal{V} \in RS(\mathbb{F}^{m \times n}) $, then $ \mathcal{V}^\perp \in RS(\mathbb{F}^{m \times n}) $. More concretely, for any subspace $ \mathcal{L} \subseteq \mathbb{F}^n $, it holds that
$$ (\mathcal{V}_\mathcal{L})^\perp = \mathcal{V}_{(\mathcal{L}^\perp)}. $$
\end{lemma}

\begin{lemma} [\textbf{Forney's duality \cite{forney}}] \label{forney lemma}
Given vector spaces $ \mathcal{C}, \mathcal{V} \subseteq \mathbb{F}^{m \times n} $, it holds that
\begin{equation*}
\dim(\mathcal{V}) - \dim((\mathcal{C}^\perp) \cap \mathcal{V}) = \dim(\mathcal{C}) - \dim(\mathcal{C} \cap (\mathcal{V}^\perp)).
%\label{forney equation}
\end{equation*}
\end{lemma}

{\color{black}We now show that all GMWs of a linear code determine uniquely those of the corresponding dual code. Since GMWs and DGWs \cite{ravagnaniweights} coincide when $ \mathbb{F} = \mathbb{F}_q $ and $ m \neq n $ by Theorem \ref{th GMW extend DGW}, the next result coincides with \cite[Corollary 38]{ravagnaniweights} in such cases:}

\begin{proposition} \label{prop duality theorem}
Given a linear code $ \mathcal{C} \subseteq \mathbb{F}^{m \times n} $ with $ k = \dim(\mathcal{C}) $, and given an integer $ p \in \mathbb{Z} $, define
\begin{equation*}
\begin{split}
W_p(\mathcal{C}) = & \{ d_{M,p + rm}(\mathcal{C}) \mid r \in \mathbb{Z}, 1 \leq p+rm \leq k \}, \\
\overline{W}_p(\mathcal{C}) = & \{ n + 1 - d_{M,p + rm}(\mathcal{C}) \mid r \in \mathbb{Z}, 1 \leq p+rm \leq k \}.
\end{split}
\end{equation*}
Then it holds that
$$ \{ 1,2, \ldots, n \} = W_p(\mathcal{C}^\perp) \cup \overline{W}_{p + k}(\mathcal{C}), $$
where the union is disjoint.
\end{proposition}

{\color{black}The proof of this proposition can be translated word by word from the proof of \cite[Corollary 38]{ravagnaniweights} using the monotonicity properties from Proposition \ref{monotonicity of RGMW}. However, \cite[Corollary 38]{ravagnaniweights} relies on \cite[Theorem 37]{ravagnaniweights}, and therefore we need to extend such result to the cases $ \mathbb{F} \neq \mathbb{F}_q $ or $ m = n $. The following lemma constitutes such extension:}

\begin{lemma}
Given a linear code $ \mathcal{C} \subseteq \mathbb{F}^{m \times n} $ with $ k = \dim(\mathcal{C}) $, and given $ 1 \leq r \leq k $ and $ 1 \leq s \leq mn - k $, it holds that
$$ d_{M,s} (\mathcal{C}^\perp) \neq n + 1 - d_{M,r}(\mathcal{C}) $$
if $ r = p + k + r^\prime m $ and $ s = p + s^\prime m $, for some integers $ p, r^\prime, s^\prime \in \mathbb{Z} $.
\end{lemma} 
\begin{proof}
Assume that equality holds for a pair of such $ r $ and $ s $. {\color{black}Denote $ \mathcal{C}_\mathcal{L} = \mathcal{C} \cap \mathcal{V}_\mathcal{L} $, for a linear subspace $ \mathcal{L} \subseteq \mathbb{F}^n $, and rewrite} Proposition \ref{connection between RDIP RGMW} as follows:
\begin{equation}
\begin{split}
d_{M,r}(\mathcal{C}) = \min \{ \mu \mid \max \{ & \dim(\mathcal{C}_\mathcal{L}) \mid \mathcal{L} \subseteq \mathbb{F}^n, \\
 & \dim(\mathcal{L}) = \mu \} \geq r \}.
\end{split}
\label{connection equation}
\end{equation}

Write $ d_{M,r}(\mathcal{C}) = \mu $. Then Equation (\ref{connection equation}) implies that
\begin{equation}
\max \{ \dim(\mathcal{C}_\mathcal{L}) \mid \mathcal{L} \subseteq \mathbb{F}^n, \dim(\mathcal{L}) = \mu \} \geq r,
\label{duality theorem proof 1}
\end{equation}
and $ \mu $ is the minimum integer with such property. Now write $ d_{M,s}(\mathcal{C}^\perp) = \nu = n+1-\mu $. In the same way, Equation (\ref{connection equation}) implies that
$$ \max \{ \dim((\mathcal{C}^\perp)_\mathcal{L}) \mid \mathcal{L} \subseteq \mathbb{F}^n, \dim(\mathcal{L}) = \nu \} \geq s. $$
On the other hand, given a subspace $ \mathcal{L} \subseteq \mathbb{F}^n $ with $ \dim(\mathcal{L}) = \nu $, we have that
$$ \dim(\mathcal{C}_{\mathcal{L}^\perp}) = \dim(\mathcal{C} \cap (\mathcal{V}_\mathcal{L})^\perp) = k - m \nu + \dim((\mathcal{C}^\perp)_\mathcal{L}), $$
where the first equality follows from Lemma \ref{lemma dual rank support}, and the second equality follows from Lemma \ref{forney lemma} and Equation (\ref{dimension matrix modules}). Therefore, it holds that
\begin{equation}
\begin{split}
& \max \{ \dim(\mathcal{C}_\mathcal{L}) \mid \mathcal{L} \subseteq \mathbb{F}^n, \dim(\mathcal{L}) = \mu - 1 \} \\
& \geq k - m \nu + s = k - mn - m + m \mu + s.
\end{split}
\label{duality theorem proof 2}
\end{equation}
From the fact that $ \mu $ is the minimum integer satisfying Equation (\ref{duality theorem proof 1}), and from Equation (\ref{duality theorem proof 2}), we conclude that
$$ k - mn - m + m \mu + s < r. $$
Now if we interchange the roles of $ \mathcal{C} $ and $ \mathcal{C}^\perp $, and the roles of $ r $ and $ s $, then we automatically interchange the roles of $ \mu $ and $ n+1-\mu $, and the roles of $ k $ and $ mn-k $. Therefore, we may also conclude that
$$ k - mn + m \mu + s > r. $$
Using the expressions $ r = p + k + r^\prime m $ and $ s = p + s^\prime m $, and dividing everything by $ m $, the previous two inequalities are, respectively
$$ s^\prime - n - 1 + \mu < r^\prime, \quad \textrm{and} \quad s^\prime - n + \mu > r^\prime, $$
which contradict each other. Hence the lemma follows.
\end{proof}

{\color{black}Observe} that the duality theorem for GRWs \cite{jerome} is a direct consequence of Theorem \ref{th RGMW extend RGRW} and Proposition \ref{prop duality theorem}:

\begin{corollary}[\textbf{\cite{jerome}}] \label{corollary duality grw}
Given an $ \mathbb{F}_{q^m} $-linear code $ \mathcal{C} \subseteq \mathbb{F}_{q^m}^n $ of dimension $ k $ over $ \mathbb{F}_{q^m} $, denote $ d_r = d_{R,r}(\mathcal{C}) $ and $ d_s^\perp = d_{R,s}(\mathcal{C}^\perp) $, for $ 1 \leq r \leq k $ and $ 1 \leq s \leq n-k $. Then
$$ \{ 1,2, \ldots, n \} = \{ d_1, d_2, \ldots, d_k \} $$
$$ \cup \{ n+1 - d_1^\perp, n+1 - d_2^\perp, \ldots, n+1 - d_{n-k}^\perp \}, $$
where the union is disjoint.
\end{corollary}

Finally, we show that the duality theorem for GHWs \cite{wei} is a consequence of {\color{black}Theorem \ref{th RGMW extend RGHW}} and Proposition \ref{prop duality theorem}:

\begin{corollary}[\textbf{\cite{wei}}] \label{corollary duality wei}
Given a linear code $ \mathcal{C} \subseteq \mathbb{F}^n $ of dimension $ k $, denote $ d_r = d_{H,r}(\mathcal{C}) $ and $ d_s^\perp = d_{H,s}(\mathcal{C}^\perp) $, for $ 1 \leq r \leq k $ and $ 1 \leq s \leq n-k $. Then
$$ \{ 1,2, \ldots, n \} = \{ d_1, d_2, \ldots, d_k \} $$
$$ \cup \{ n+1 - d_1^\perp, n+1 - d_2^\perp, \ldots, n+1 - d_{n-k}^\perp \}, $$
where the union is disjoint.
\end{corollary}
\begin{proof}
We will use the notation in Proposition \ref{prop duality theorem} during the whole proof. First of all, by {\color{black}Theorem \ref{th RGMW extend RGHW}} it holds that $ W_p(\Delta(\mathcal{C})) = \{ d_{H,p}(\mathcal{C}) \} $ if $ 1 \leq p\bmod{n} \leq k $ and $ W_p(\Delta(\mathcal{C})) = \emptyset $ if $ k+1 \leq p\bmod{n} \leq n-1 $ or $ p \bmod{n} = 0 $. Therefore
$$ \bigcup_{p=1}^n W_{p - k}(\Delta(\mathcal{C})) = \{ d_1, d_2, \ldots, d_k \}. $$
On the other hand, from Proposition \ref{prop duality theorem} it follows that
$$ \left( \bigcup_{p=1}^n W_{p-k}(\Delta(\mathcal{C})) \right) \cup \left( \bigcap_{p=1}^n \overline{W}_p(\Delta(\mathcal{C})^\perp) \right) = \{ 1,2, \ldots, n \}, $$
where the union is disjoint. Hence we only need to show that $ n+1 - d_s^\perp \in \overline{W}_p(\Delta(\mathcal{C})^\perp) $, for $ p = 1,2, \ldots, n $ and $ s=1,2, \ldots, n-k $.

Denote by $ \mathcal{D}_n \subseteq \mathbb{F}^{n \times n} $ the vector space of matrices with zero components in their diagonals. It holds that $ \Delta(\mathcal{C})^\perp = \Delta(\mathcal{C}^\perp) \oplus \mathcal{D}_n $.

Fix $ 1 \leq s \leq n-k $ and denote $ d = d_{H,s}(\mathcal{C}^\perp) $. First, consider a subspace $ \mathcal{D} \subseteq \mathcal{C}^\perp $ with $ {\rm wt_H}(\mathcal{D}) = d $ and $ \dim(\mathcal{D}) = s $, and define $ \mathcal{D}^\prime \subseteq \Delta(\mathcal{C})^\perp $ as the direct sum of $ \Delta(\mathcal{D}) $ and all matrices in $ \mathcal{D}_n $ with columns in the Hamming support of $ \mathcal{D} $. Since $ \dim(\mathcal{D}^\prime) = d(n-1) + s $ and $ {\rm wt_R}(\mathcal{D}^\prime) = d$, by Proposition \ref{proposition as minimum rank weights} it follows that
\begin{equation}
d_{M,d(n-1)+s}(\Delta(\mathcal{C})^\perp) \leq d.
\label{wei duality eq 1}
\end{equation}

On the other hand, assume that
$ d_{M,(d-1)(n-1)+s}(\Delta(\mathcal{C})^\perp) = d^\prime < d $. Let
$ \mathcal{E} \subseteq \Delta(\mathcal{C})^\perp $ be such that $ {\rm wt_R}(\mathcal{E}) = d^\prime $ and $ \dim(\mathcal{E}) = (d-1)(n-1) + s $. Denote by $ \mathcal{E}_D $ the vector space of matrices obtained by replacing the elements outside the diagonal of those matrices in $ \mathcal{E} $ by zero. If $ \mathcal{L} = {\rm RSupp}(\mathcal{E}) \subseteq \mathbb{F}^n $, we claim that
\begin{equation}
  \dim(\mathcal{E} \cap \mathcal{D}_n) \leq n \mathrm{wt_R}(\mathcal{E})
  - {\rm wt_H}(\mathcal{L}). \label{eqaa}
\end{equation}
It is sufficient to show that $ \dim (\mathcal{V}_\mathcal{L} \cap \mathcal{D}_n) = n \dim(\mathcal{L}) - {\rm wt_H}(\mathcal{L}) $. Denote by $ \mathcal{V}_{\mathcal{L} D} $ the vector space of matrices obtained by replacing the elements outside the diagonal of those matrices in $ \mathcal{V}_\mathcal{L} $ by zero. Then, by Proposition \ref{prop characterization},
$ \dim(\mathcal{V}_{\mathcal{L}D}) = {\rm wt_H}(\mathcal{L}) $, and $ \dim( \mathcal{V}_\mathcal{L} \cap \mathcal{D}_n) = \dim(\mathcal{V}_\mathcal{L}) - \dim( \mathcal{V}_{\mathcal{L}D}) = n \dim (\mathcal{L}) - {\rm wt_H}(\mathcal{L}) $.

By monotonicity (Proposition \ref{monotonicity of RGMW}), we have that $ d^\prime = d-1 $, and thus $\dim( \mathcal{E}) = d^\prime(n-1)+s$. Therefore, by (\ref{eqaa}), $ \dim(\Delta^{-1}(\mathcal{E}_D)) = \dim (\mathcal{E}_D)
= \dim (\mathcal{E}) - \dim (\mathcal{E} \cap \mathcal{D}_n) \geq s + {\rm wt_H}(\mathcal{L}) - d^\prime $. Choose indices $i_1, i_2, \ldots, i_{{\rm wt_H}(\mathcal{L}) - d'}$ from $\mathrm{HSupp}(\Delta^{-1}(\mathcal{E}_D))$,
and define
$$ \mathcal{W} = \{ \mathbf{c} \in \Delta^{-1}(\mathcal{E}_D) \mid c_{i_j} = 0, 1 \leq j \leq {\rm wt_H}(\mathcal{L}) - d^\prime \}. $$
Then $ \mathcal{W} \subseteq \mathcal{C}^\perp $, $ \dim (\mathcal{W}) \geq s $, and $ {\rm wt_H}(\mathcal{W}) \leq
{\rm wt_H}( \Delta^{-1}(\mathcal{E}_D)) - {\rm wt_H}(\mathcal{L}) + d^\prime \leq d^\prime $, which implies $ d_{H,s}(\mathcal{C}^\perp) = d^\prime < d$, which is a contradiction. Hence
\begin{equation}
d_{M,(d-1)(n-1)+s}(\Delta(\mathcal{C})^\perp) \geq d.
\label{wei duality eq 2}
\end{equation}

Combining Equation (\ref{wei duality eq 1}) and Equation (\ref{wei duality eq 2}), we conclude that
$$ d_{M,(d-1)(n-1)+s+j}(\Delta(\mathcal{C})^\perp) = d, $$
for $ j = 0,1,2, \ldots, n-1 $, which implies that $ n+1 - d_s^\perp \in \overline{W}_p(\Delta(\mathcal{C})^\perp) $, for $ p = 1,2, \ldots, n $, and we are done.
\end{proof}

{\color{black}
\section{Construction of explicit subspace designs} \label{app explicit designs}

In this appendix, we recall how to construct the subspace design formed by $ \mathcal{H}_0, \mathcal{H}_1, \mathcal{H}_2, \ldots \subseteq  \mathbb{F}_{q^m} $ in Section \ref{sec list decoding}. This construction is given in \cite{list-decodable-rank-metric}, based on a construction in \cite{guruswami-designs}, and is explicit in the sense that it can be constructed using an algorithm of polynomial complexity on $ q $.

Fix $ \varepsilon > 0 $ and a positive integer $ s $ such that $ 4 s n \leq \varepsilon m $, and assume that $ n $ divides $ m $. Let $ d_1=q^{m/n-1}, d_2 = q^{m/n-2}, \ldots, d_{m/n}=1 $ and let $\gamma_1, \gamma_2, \ldots, \gamma_{m/n} $ be distinct non-zero elements of $\mathbb{F}_{q^n}$. Define
$$ f_i(x_1, x_2 \ldots, x_{m/n}) = \sum_{j=1}^{m/n} \gamma_j^i x_j^{d_j}, $$
for $ i = 1, 2, \ldots, s $, and let $ \mathcal{S} \subseteq \mathbb{F}_{q^n}^{m/n}$ be the set of common zeros of
$ f_1, f_2, \ldots, f_s $, which is an $ \mathbb{F}_q $-linear vector space. We may assume that $ \mathcal{S} \subseteq \mathbb{F}_{q^m} $ by an $ \mathbb{F}_{q^n} $-linear vector space isomorphism $ \mathbb{F}_{q^n}^{m / n} \cong \mathbb{F}_{q^m} $ (any isomorphism works).

Let $\beta$ be a primitive element of $\mathbb{F}_{q^n}$. For $\alpha \in \mathbb{F}_{q^{n\lfloor \frac{\varepsilon m}{2ns}\rfloor}}$, let 
$$ \mathcal{S}_\alpha = \left\lbrace \alpha^{q^j} \beta^i \mid 0 \leq j < \left\lfloor \frac{\varepsilon m}{2ns} \right\rfloor, 0 \leq i < 2s \right\rbrace . $$
The algorithm in \cite[Subsection 4.3]{guruswami-designs} gives in polynomial time over $ q $ a set $\mathcal{F} \subseteq \mathbb{F}_{q^{n\lfloor \frac{\varepsilon m}{2ns}\rfloor}}$ of size $ q^{\Omega(\frac{\varepsilon m}{n s})} $ such that:
\begin{enumerate}
\item 
$\mathbb{F}_q(\alpha) = \mathbb{F}_{q^{n\lfloor \frac{\varepsilon m}{2ns}\rfloor}}$, for all $ \alpha \in \mathcal{F} $,
\item 
$ \mathcal{S}_\alpha \cap \mathcal{S}_\beta = \emptyset $, for all distinct $ \alpha, \beta \in \mathcal{F} $, and
\item 
$ | \mathcal{S}_\alpha | = 2s \lfloor \frac{\varepsilon m}{2ns}\rfloor $, for all $ \alpha \in \mathcal{F} $.
\end{enumerate}

Define the $ \mathbb{F}_{q^n} $-linear vector space $ \mathcal{V}_\alpha \subseteq \mathbb{F}_{q^n}^{m/n} $ as
\begin{equation*}
\begin{split}
\mathcal{V}_\alpha = \{ & (a_0, a_1, \ldots, a_{m/n-1}) \in \mathbb{F}_{q^n}^{m/n}\mid \\
 & \sum_{i=0}^{m/n-1} a_i (\alpha \beta^j)^i = 0 \mid 0\leq j < 2s \},
\end{split}
\end{equation*}
for every $\alpha \in \mathcal{F}$, where we may consider $ \mathcal{V}_\alpha \subseteq \mathbb{F}_{q^m} $ as before.

Finally, the $ \mathbb{F}_q $-linear vector spaces $ \mathcal{H}_0, \mathcal{H}_1, \mathcal{H}_2, \ldots \subseteq  \mathbb{F}_{q^m} $ in Section \ref{sec list decoding} are defined as $ \mathcal{H}_i = \mathcal{S} \cap \mathcal{V}_{\alpha_i} $, for distinct $ \alpha_i \in \mathcal{F} $.

The constructions of $ \mathcal{F} $ and $ \mathcal{V}_\alpha $ appeared first in \cite[Subsection 4.2]{guruswami-designs} and $ \mathcal{S} $ appeared first in \cite[Corollary 6]{list-decodable-rank-metric}.

We conclude the appendix by computing the dimensions of the vector spaces $ \mathcal{H}_0, \mathcal{H}_1, \mathcal{H}_2, \ldots \subseteq  \mathbb{F}_{q^m} $, which is done in the proof of \cite[Theorem 8]{list-decodable-rank-metric}:

\begin{lemma}[\textbf{\cite{list-decodable-rank-metric}}] \label{lemma dimensions of Hi}
The vector spaces $ \mathcal{H}_0, \mathcal{H}_1, \mathcal{H}_2, \ldots \subseteq  \mathbb{F}_{q^m} $ have dimension at least $ m (1 - 2 \varepsilon) $ over $ \mathbb{F}_q $.
\end{lemma}

}

{\color{black}
\section{Proof of Theorem \ref{theorem characterization}} \label{app 1}

In this appendix, we give the proof of Theorem \ref{theorem characterization}, which we now recall:

\begin{reptheorem}{theorem characterization}
Let $ \phi : \mathcal{V} \longrightarrow \mathcal{W} $ be a vector space isomorphism between rank support spaces $ \mathcal{V} \in RS(\mathbb{F}^{m \times n}) $ and $ \mathcal{W} \in RS(\mathbb{F}^{m \times n^\prime}) $, and consider the following properties:
\begin{itemize}
\item[(P 1)] There exist full-rank matrices $ A \in \mathbb{F}^{m \times m} $ and $ B \in \mathbb{F}^{n \times n^\prime} $ such that $ \phi(C) = ACB $, for all $ C \in\mathcal{V} $.
\item[(P 2)] A subspace $ \mathcal{U} \subseteq \mathcal{V} $ is a rank support space if, and only if, $ \phi(\mathcal{U}) $ is a rank support space. 
\item[(P 3)] For all subspaces $ \mathcal{D} \subseteq \mathcal{V} $, it holds that $ {\rm wt_R}(\phi(\mathcal{D})) = {\rm wt_R}(\mathcal{D}) $.
\item[(P 4)] $ \phi $ is a rank isometry.
\end{itemize}
Then the following implications hold:
$$ (\textrm{P 1}) \Longleftrightarrow (\textrm{P 2}) \Longleftrightarrow (\textrm{P 3}) \Longrightarrow (\textrm{P 4}). $$
In particular, a security equivalence is a rank isometry and, in the case $ \mathcal{V} = \mathcal{W} = \mathbb{F}^{m \times n} $ and $ m \neq n $, the reversed implication holds by Proposition \ref{morrison proposition}.
\end{reptheorem} 
}
\begin{proof} 
First we prove $ (\textrm{P 1}) \Longrightarrow (\textrm{P 2}) $: It follows immediately from the characterization of rank support spaces in Proposition \ref{prop characterization}, item 3.

Now we prove $ (\textrm{P 2}) \Longrightarrow (\textrm{P 3}) $: Let $ \mathcal{L} = {\rm RSupp}(\mathcal{D}) \subseteq \mathbb{F}^n $ and $ \mathcal{L}^\prime = {\rm RSupp}(\phi(\mathcal{D})) \subseteq \mathbb{F}^{n^\prime} $. It holds that $ \mathcal{V}_\mathcal{L} \subseteq \mathcal{V} $ and $ \mathcal{V}_{\mathcal{L}^\prime} \subseteq \mathcal{W} $, and they are the smallest rank support spaces in $ \mathcal{V} $ and $ \mathcal{W} $ containing $ \mathcal{D} $ and $ \phi(\mathcal{D}) $, respectively, by Lemma \ref{basic lemma matrix modules}. Since $ \phi $ preserves rank support spaces and their inclusions, we conclude that $ \phi(\mathcal{V}_\mathcal{L}) = \mathcal{V}_{\mathcal{L}^\prime} $, which implies that $ \dim(\mathcal{L}) = \dim(\mathcal{L}^\prime) $ by (\ref{dimension matrix modules}), and (P 3) follows.

Next we prove $ (\textrm{P 2}) \Longleftarrow (\textrm{P 3}) $: Assume that $ \mathcal{U} \subseteq \mathcal{V} $ is a rank support space. This means that $ m{\rm wt_R}(\mathcal{U}) = \dim(\mathcal{U}) $ by (\ref{dimension matrix modules}). Since $ \phi $ satisfies (P 3) and is a vector space isomorphism, we conclude that $ m{\rm wt_R}(\phi(\mathcal{U})) = \dim(\phi(\mathcal{U})) $, and thus $ \phi(\mathcal{U}) $ is a rank support space also by (\ref{dimension matrix modules}). Similarly we may prove that, if $ \phi(\mathcal{U}) $ is a rank support space, then $ \mathcal{U} $ is a rank support space.

Now we prove $ (\textrm{P 3}) \Longrightarrow (\textrm{P 4}) $: Trivial from the fact that $ {\rm wt_R}(\langle \{ C \} \rangle) = {\rm Rk}(C) $, for all $ C \in \mathcal{V} $.

Finally we prove $ (\textrm{P 1}) \Longleftarrow (\textrm{P 2}) $: Denote $ \dim(\mathcal{V}) = \dim(\mathcal{W})= mk $ and consider bases of $ \mathcal{V} $ and $ \mathcal{W} $ as in Proposition \ref{prop characterization}, item 2. By defining vector space isomorphisms $ \mathbb{F}^{m \times k} \longrightarrow \mathcal{V} $ and $ \mathcal{W} \longrightarrow \mathbb{F}^{m \times k} $, sending such bases to the canonical basis of $ \mathbb{F}^{m \times k} $, we see that we only need to prove the result for the particular case $ \mathcal{V} = \mathcal{W} = \mathbb{F}^{m \times n} $.

Denote by $ E_{i,j} \in \mathbb{F}^{m \times n} $ the matrices in the canonical basis, for $ 1 \leq i \leq m $, $ 1 \leq j \leq n $, that is, $ E_{i,j} $ has $ 1 $ in its $ (i,j) $-th component, and zeroes in its other components. 

Consider the rank support space $ \mathcal{U}_j = \langle E_{1,j}, E_{2,j}, \ldots, E_{m,j} \rangle \subseteq \mathbb{F}^{m \times n} $, for $ 1 \leq j \leq n $. Since $ \phi(\mathcal{U}_j) $ is a rank support space, it has a basis $ B_{i,j} $, $ i = 1,2, \ldots, m $, as in Proposition \ref{prop characterization}, item 2, for a vector $ \mathbf{b}_j \in \mathbb{F}^n $. This means that
$$ \phi(E_{i,j}) = \sum_{s=1}^m a_{s,i}^{(j)} B_{s,j}, $$
for some $ a_{s,i}^{(j)} \in \mathbb{F} $, for all $ s,i = 1,2, \ldots, m $ and $ j = 1,2, \ldots, n $. If we define the matrix $ A^{(j)} \in \mathbb{F}^{m \times m} $ whose $ (s,i) $-th component is $ a_{s,i}^{(j)} $, and $ B \in \mathbb{F}^{n \times n} $ whose $ j $-th row is $ \mathbf{b}_j $, then a simple calculation shows that
$$ \phi(E_{i,j}) = A^{(j)} E_{i,j} B, $$
and the matrices $ A^{(j)} $ and $ B $ are invertible. If we prove that there exist non-zero $ \lambda_j \in \mathbb{F} $ with $ A^{(j)} = \lambda_j A^{(1)} $, for $ j = 2,3, \ldots, n $, then we are done, since we can take the vectors $ \lambda_j \mathbf{b}_j $ instead of $ \mathbf{b}_j $, define $ A = A^{(1)} $, and then it holds that
$$ \phi(E_{i,j}) = A E_{i,j} B, $$
for all $ i = 1,2, \ldots, m $ and $ j = 1,2, \ldots, n $, implying (P 1).

To this end, we first denote by $ \mathbf{a}_i^{(j)} \in \mathbb{F}^m $ the $ i $-th column in $ A^{(j)} $ (written as a row vector). Observe that we have already proven that $ \phi $ preserves ranks. Hence $ {\rm Rk}(\phi(E_{i,j} + E_{i,1})) = 1 $, which means that $ {\rm Rk}(A^{(j)} E_{i,j} + A^{(1)} E_{i,1}) = 1 $, which implies that there exist $ \lambda_{i,j} \in \mathbb{F} $ with
$$ \mathbf{a}_i^{(j)} = \lambda_{i,j} \mathbf{a}_i^{(1)}. $$
On the other hand, a matrix calculation shows that
$$ \phi \left( \sum_{i=1}^m \sum_{j=1}^n E_{i,j} \right) = \left( \sum_{i=1}^m \mathbf{a}_i^{(1)}, \sum_{i=1}^m \mathbf{a}_i^{(2)}, \ldots, \sum_{i=1}^m \mathbf{a}_i^{(n)} \right) B $$
$$ = \left( \sum_{i=1}^m \mathbf{a}_i^{(1)}, \sum_{i=1}^m \lambda_{i,2} \mathbf{a}_i^{(1)}, \ldots, \sum_{i=1}^m \lambda_{i,n} \mathbf{a}_i^{(1)} \right) B. $$
Since $ {\rm Rk}(\sum_{i=1}^m \sum_{j=1}^n E_{i,j}) = 1 $ and the vectors $ \mathbf{a}_i^{(1)} $, $ 1 \leq i \leq m $, are linearly independent, we conclude that $ \lambda_{i,j} $ depends only on $ j $ and not on $ i $, and we are done.
\end{proof}

\section{Matrix modules} \label{app matrix modules}

Rank support spaces can also be seen as left submodules of the left module $ \mathbb{F}^{m \times n} $ over the (non-commutative) ring $ \mathbb{F}^{m \times m} $. This has been used in Example \ref{example delsarte less than}. Since we think this result is of interest by itself, we include the characterization in this appendix.

\begin{definition} [\textbf{Matrix modules}]
We say that a set $ \mathcal{V} \subseteq \mathbb{F}^{m \times n} $ is a matrix module if 
\begin{enumerate}
\item
$ V + W \in \mathcal{V} $, for every $ V, W \in \mathcal{V} $, and
\item
$ M V \in \mathcal{V} $, for every $ M \in \mathbb{F}^{m \times m} $ and every $ V \in \mathcal{V} $.
\end{enumerate}
\end{definition}

\begin{proposition} 
A set $ \mathcal{V} \subseteq \mathbb{F}^{m \times n} $ is a rank support space if, and only if, it is a matrix module.
\end{proposition}
\begin{proof}
Assume that $ \mathcal{V} $ is a rank support space. Using the characterization in Proposition \ref{prop characterization}, item 3, it is trivial to see that $ \mathcal{V} $ is a matrix module.

Assume now that $ \mathcal{V} $ is a matrix module. It holds that $ \mathcal{V} $ is a vector space. Let $ \mathcal{L} = {\rm RSupp}(\mathcal{V}) $, and take $ \mathbf{v} \in \mathcal{L} $. There exist $ V_1, V_2, \ldots, V_s \in \mathcal{V} $ and $ \mathbf{v}_j \in {\rm Row}(V_j) $, for $ j = 1,2, \ldots, s $, such that $ \mathbf{v} = \sum_{j=1}^s \mathbf{v}_j $.

For fixed $ 1 \leq i \leq m $ and $ 1 \leq j \leq s $, it is well-known that there exists $ M_{i,j} \in \mathbb{F}^{m \times m} $ such that $ M_{i,j}V_j $ has $ \mathbf{v}_j $ as its $ i $-th row and the rest of its rows are zero vectors. Since $ \mathcal{V} $ is closed under sums of matrices, we conclude that $ \mathcal{V}_\mathcal{L} \subseteq \mathcal{V} $ and therefore both are equal.
\end{proof}

\appendices

\section*{Acknowledgment}

The authors wish to thank Alberto Ravagnani for clarifying the relation between GMWs and DGWs. The first author is also thankful for the support and guidance of his advisors Olav Geil and Diego Ruano. At the time of submission, the first author was visiting the Edward S. Rogers Sr. Department of Electrical and Computer Engineering, University of Toronto. He greatly appreciates the support and hospitality of Frank R. Kschischang. {\color{black}Finally, the authors also wish to thank the editor and anonymous reviewers for their very helpful comments.}

\ifCLASSOPTIONcaptionsoff
  \newpage
\fi



\bibliographystyle{IEEEtranS}
\end{document}